\documentclass[a4paper]{article}

\usepackage{fullpage}

\usepackage{color}
\usepackage{graphicx}
\usepackage{amsmath}
\usepackage{amssymb}
\usepackage{amsthm}
\usepackage[noend]{algpseudocode}
\usepackage{algorithm}
\usepackage{tikz}
\usepackage{enumitem}
\usepackage{booktabs}
\usepackage{tabularx}
\usepackage{array}
\usepackage[hidelinks]{hyperref}
\usepackage{cleveref}

\theoremstyle{plain}
\newtheorem{theorem}{Theorem}
\newtheorem{lemma}[theorem]{Lemma}
\newtheorem{proposition}[theorem]{Proposition}
\newtheorem{corollary}[theorem]{Corollary}

\theoremstyle{definition}
\newtheorem{definition}[theorem]{Definition}
\newtheorem{assumption}[theorem]{Assumption}

\theoremstyle{remark}
\newtheorem{remark}[theorem]{Remark}

\let\geq\geqslant
\let\leq\leqslant
\def\eps{\varepsilon}
\newcommand{\IR}{\mathbb{R}}

\newcommand{\IP}{\mathbb{P}}
\newcommand{\NVD}{\operatorname{NVD}}
\newcommand{\FVD}{\operatorname{FVD}}

\title{Quadratic Complexity of Voronoi Diagrams in $\mathbb{R}^3$ for Lines in a Single Ruling of a Regulus
\thanks{
This work was supported by National Research Foundation of Korea (NRF) grants funded by the Korea government (MSIT) (No. RS-2026-25471649 and No. RS-2024-00414849).
}}

\author{Eunku Park\\
Department of Liberal arts and Sciences\\
DGIST, Republic of Korea\\
\texttt{parkeun9@dgist.ac.kr}
\and
Kangjin Han\\ 
Department of Liberal arts and Sciences\\ 
DGIST, Republic of Korea\\
\texttt{kjhan@dgist.ac.kr}
}

\begin{document}
\maketitle

\begin{abstract}
We study nearest and farthest Voronoi diagrams of lines in three-dimensional Euclidean space, which partition space according to the closest and farthest input lines, respectively. For arbitrary lines, a gap remains between the known quadratic lower bound and near-cubic upper bounds for the complexity of the nearest diagram. We consider input lines chosen from one of two continuous families of straight lines on a smooth surface defined by a quadratic equation, with one line from each family passing through every point of the surface.

Under the geometric assumptions stated in the paper, we prove quadratic upper bounds on the total number of vertices, edges, faces, and connected three-dimensional regions in both diagrams. We also construct quadratically many nearest vertices using lines from a single family on one fixed surface, establishing a tight quadratic worst-case bound. The farthest diagram has quadratic total complexity for every input satisfying our assumptions.

Our proof uses a parametrization that arranges the input lines in a cyclic order and expresses tangency to a sphere by a polynomial equation of degree four. At a Voronoi vertex, the four tangent input lines give four distinct roots. The sign changes between these roots, together with the nearest or farthest condition, force two intervening parameter intervals to contain no other input parameters. Thus, only quadratically many choices of four lines can define a vertex, and each choice yields at most eight isolated candidate centers.
\end{abstract}

\section{Introduction}\label{sec:introduction}
Voronoi diagrams are fundamental space-partitioning structures in computational geometry; see Aurenhammer's survey~\cite{Aurenhammer1991}. Given a finite set $L$ of $n$ lines in $\IR^3$, called sites, the nearest and farthest Voronoi diagrams subdivide space according to the lines that minimize and maximize, respectively, the Euclidean distance from a point. We denote these diagrams by $\NVD(L)$ and $\FVD(L)$. Their combinatorial complexity is the total number of vertices, edges, two-dimensional faces, and three-dimensional cells. Here, three-dimensional cells are counted as connected regions, since the region associated with a single farthest site may have several connected components.

The combinatorial complexity of the nearest Voronoi diagram of arbitrary lines in $\IR^3$ is a long-standing open problem. Aronov constructed instances of complexity $\Omega(n^2)$~\cite{Aronov2002}, whereas the best known upper bound remains $O(n^{3+\eps})$ for every fixed $\eps>0$~\cite{Sharir1994,AgarwalAronovSharir1997}. The upper bound follows from results on the pointwise minimum of the distance functions, also called their lower envelope. The gap between quadratic and near-cubic complexity motivates the study of geometric restrictions under which a quadratic bound can be established.

We consider smooth surfaces defined by polynomial equations of degree two, called quadric surfaces, that contain two continuous one-parameter families of straight lines. Through every point of such a surface passes exactly one line from each family. The surface is called doubly ruled, and each family is called a ruling. The individual lines in a ruling are also called generators. Our input consists of lines selected from a single ruling. Distinct lines in that ruling are pairwise skew, meaning that they neither intersect nor are parallel. Their directions vary continuously, so the family is not restricted to a fixed finite number of orientations.

Figure~\ref{fig:doubly-ruled-quadrics} illustrates the two rulings on standard examples. Explicit parametrizations and verifications of their incidence properties are given in Appendix~\ref{app:explicit-rulings}.

\begin{figure}
	\centering
	\includegraphics[width=5in]{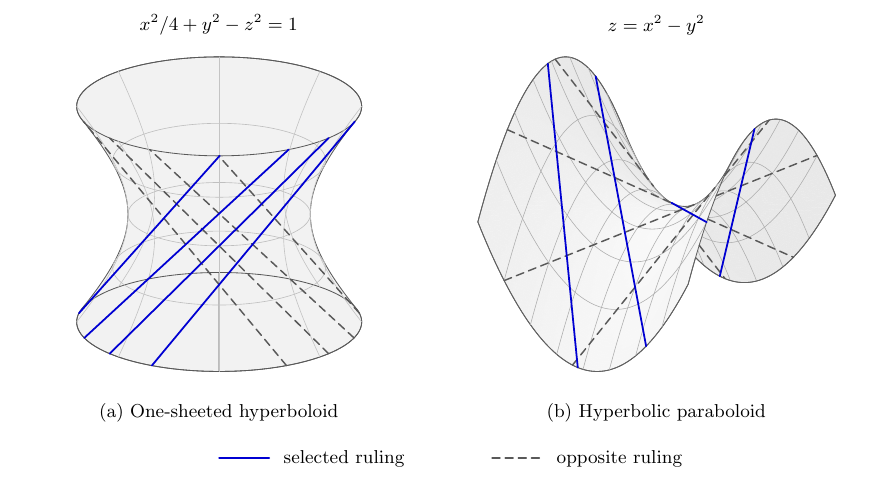}
	\caption{Two smooth doubly ruled quadrics: (a) the non-rotational one-sheeted hyperboloid $x^2/4+y^2-z^2=1$ and (b) the hyperbolic paraboloid $z=x^2-y^2$. Solid blue segments represent input lines selected from one ruling, while dashed gray segments represent lines from the opposite ruling. Distinct lines within each ruling are pairwise skew. The hyperboloid admits input sets satisfying Assumption~\ref{ass:gp}, as established in Section~\ref{sec:lower-bound}.
In contrast, any three distinct lines from a single ruling of the hyperbolic paraboloid are parallel to a common plane, violating Assumption~\ref{ass:gp}(1).}
	\label{fig:doubly-ruled-quadrics}
\end{figure}

Our results require the geometric assumptions collected in Assumption~\ref{ass:gp}. These exclude three lines parallel to a common plane, a sphere tangent to four sites at coplanar contact points, and a sphere tangent to five sites. They also require that no choice of unit direction vectors for four distinct sites gives four cocircular points on the unit sphere; a unit direction vector is a vector of length one parallel to the corresponding line. In addition, no Euclidean sphere may be tangent to every line of the chosen ruling. We refer to inputs satisfying these conditions as being in general position.

The surfaces in Figure~\ref{fig:doubly-ruled-quadrics} illustrate the two rulings of smooth quadric surfaces. The non-rotational hyperboloid in the left panel is used in our lower-bound construction, which establishes that one ruling of this surface contains input sets satisfying Assumption~\ref{ass:gp}. For the hyperbolic paraboloid in the right panel, all lines within either ruling are parallel to a common plane. Hence, any three sites chosen from the same ruling violate Assumption~\ref{ass:gp}(1).

\subparagraph{Main result.}
We prove quadratic upper bounds for both diagrams and a worst-case lower bound for the nearest diagram. The farthest diagram has quadratic total complexity for every input satisfying our assumptions.

\begin{theorem}\label{thm:main}
Let $L$ be a set of $n\ge 4$ distinct lines selected from one ruling of a smooth doubly ruled real quadric and satisfying Assumption~\ref{ass:gp}. Let $V_N$ and $V_F$ denote the numbers of vertices of $\NVD(L)$ and $\FVD(L)$, respectively. Then the following hold.
\begin{enumerate}
	\item The nearest Voronoi diagram has at most
	\[
		V_N\le 4n(n-3)
	\]
	vertices and $O(n^2)$ total combinatorial complexity.

	\item The farthest Voronoi diagram has at most
	\[
		V_F\le 4n(n-3)
	\]
	vertices and $\Theta(n^2)$ total combinatorial complexity.

	\item For every $n\ge 4$, one ruling of the fixed surface
	\[
		Q_*:\frac{x^2}{4}+y^2-z^2=1
	\]
	contains a set of $n$ lines satisfying Assumption~\ref{ass:gp} whose nearest Voronoi diagram has at least
	\[
		\frac{(n-2)(n-3)}{2}
	\]
	distinct vertices.
\end{enumerate}
Consequently, the worst-case combinatorial complexity of the nearest Voronoi diagram in this class is $\Theta(n^2)$.
\end{theorem}

The quadratic lower bound for every farthest diagram follows from its $n(n-1)$ three-dimensional cells~\cite[Theorem~5.12]{BarequetPapadopoulouSuderland2024}.

The proof also yields an exact quadratic-time algorithm for enumerating all finite nearest and farthest Voronoi vertices under the computational model stated with Corollary~\ref{cor:vertex-enumeration}. The algorithmic details are given in Appendix~\ref{app:algorithmic-details}.

\subparagraph{Our approach.}
The proof connects the positions of the input lines within their ruling to the signs of a polynomial describing contact with a sphere. A parametrization assigns a parameter value to each line in the family. We use a parameter space with a circular order, so the input lines inherit a cyclic order. A cyclic gap is an open interval between consecutive input parameters in this order.

For a sphere with center $c$ and squared radius $R$, the sign of $d(c,\ell)^2-R$ describes the position of a line $\ell$ relative to the corresponding ball. A negative value means that the line meets the open ball, zero means tangency, and a positive value means that the line misses the closed ball. We express this sign by a polynomial of degree four in the parameters of the ruling, which we call the contact polynomial. The assumption excluding a sphere tangent to the entire ruling ensures that this polynomial is not identically zero.

At a Voronoi vertex, exactly four input lines are tangent to the sphere centered at that vertex whose radius is their common distance. Their four distinct parameters account for all four roots of the contact polynomial. Each root therefore has multiplicity one, and the sign changes whenever a root is crossed.

For a nearest vertex, the open ball meets no input line. Hence, the two intervals on which the contact polynomial is negative contain no input parameters. These intervals must be two cyclic gaps with no common endpoint. For a farthest vertex, the closed ball meets every input line, and the same conclusion holds for the two positive intervals. Among the $n$ cyclic gaps, the number of unordered pairs with no common endpoint is
\[
	\binom{n}{2}-n=\frac{n(n-3)}{2}.
\]
Each pair specifies four input lines. Their common equidistant centers satisfy three quadratic equations in three variables. The bound on isolated solutions given by B\'ezout's inequality is the product of these degrees, namely eight. This proves the vertex bounds in Theorem~\ref{thm:main}. The relations between feature counts established by Papadopoulou and Wang~\cite[Theorem~9]{PapadopoulouWang2026} then give the bounds on total combinatorial complexity.

The coordinate representation explains why tangency between a sphere and a line in the ruling can be described by a polynomial of degree four. We represent each line by six coordinates, where multiplying all six coordinates by the same nonzero scalar does not change the represented line. The points representing the selected ruling lie on a quadratic curve in a plane. To describe this curve, we use pairs of real parameters that are not both zero, with pairs differing by a common nonzero factor representing the same parameter value. These parameter values have a circular order~\cite[Example~7.12]{Tu2011}. Each of the six line coordinates can be chosen as a polynomial in the two parameters in which every term has total degree two.

For a fixed sphere, tangency is expressed by a quadratic expression in these six line coordinates. Substituting their parameter expressions therefore gives a polynomial in which every term has total degree four. We call this the contact polynomial. The assumption excluding a sphere tangent to the entire ruling ensures that it is not identically zero. Multiplying both parameters by the same nonzero factor multiplies the polynomial by the fourth power of that factor, so its sign is unchanged. At every parameter representing a line in Euclidean space, the polynomial is negative when the line meets the open ball bounded by the sphere, zero when the line is tangent to the sphere, and positive when the line misses the closed ball. Figure~\ref{fig:geometric-pipeline} illustrates this representation. All distances remain Euclidean, and the required coordinate representation is developed in Section~\ref{sec:preliminaries}.

\begin{figure}
\centering
\includegraphics[width=\linewidth]{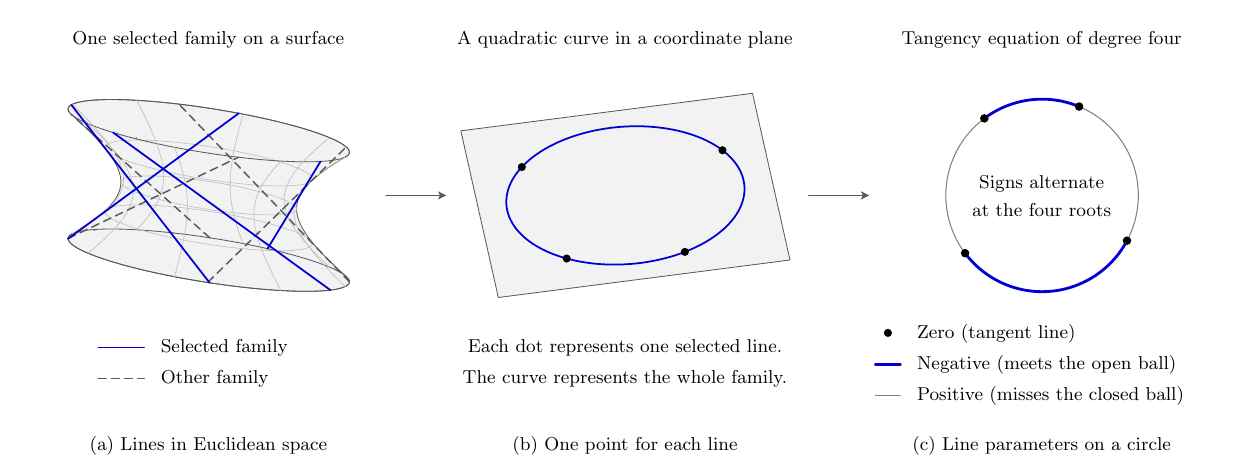}
\caption{From a ruling of lines to a contact polynomial for a fixed sphere. (a) Solid blue lines are selected from one ruling, while dashed gray lines illustrate the opposite ruling. (b) The selected ruling is represented by a planar conic, with each marked point corresponding to one selected line. (c) Substituting a parametrization of this curve into the sphere--line tangency equation yields a polynomial of degree four. In the illustrated case, the polynomial has four distinct roots on the real projective parameter circle, corresponding to four tangent lines. Its sign alternates between consecutive roots: blue arcs represent lines meeting the open ball bounded by the sphere, while gray arcs represent lines missing the closed ball.}
\label{fig:geometric-pipeline}
\end{figure}

For the nearest lower bound, we show that quadratically many gap pairs allowed by the upper-bound argument can be realized by spheres whose open balls meet no input line. We work with one ruling of the fixed surface $Q_*$ from Theorem~\ref{thm:main}. Near a chosen line, we describe the ruling by a real parameter $u$, with the chosen line corresponding to $u=0$. We first find a sphere whose contact polynomial is a positive multiple of $u^4$. At this sphere, all four roots coincide at zero. Our goal is to vary the sphere so that these roots become four prescribed distinct parameters near zero.

A sphere is specified by the three coordinates of its center and its squared radius. Near the chosen sphere, the leading coefficient of the contact polynomial remains positive. Dividing by this coefficient gives a polynomial with leading coefficient one and four remaining coefficients. Thus, the four parameters describing the sphere determine four polynomial coefficients. We prove that the matrix of first derivatives of this map is invertible at the chosen sphere. The inverse function theorem then guarantees that every sufficiently small change in the four coefficients is realized by a unique nearby sphere. We restrict the neighborhood so that the squared radius and the leading coefficient remain positive throughout.

Choose the input parameters in a sufficiently small interval around zero, and consider the gaps between consecutive parameters in their linear order. For each pair of gaps with no common endpoint, form the degree-four polynomial whose leading coefficient is one and whose roots are the four endpoints. Its coefficients are close to those of $u^4$, so the preceding argument gives a sphere realizing this polynomial after division by its positive leading coefficient. The polynomial is negative precisely inside the two chosen gaps and positive at every other input parameter. Consequently, the sphere is tangent to the four endpoint lines, and its open ball meets no input line.

The input parameters can be chosen within the fixed ruling so that all the geometric assumptions hold simultaneously. Under these assumptions, each constructed sphere has a nearest Voronoi vertex as its center. Different gap pairs give different sets of roots and therefore different spheres. Their centers are also distinct. Indeed, two such spheres with the same center would have the same radius, equal to the minimum distance from that center to the input lines, and would therefore be the same sphere. The quadratic number of gap pairs gives the nearest-vertex lower bound.

\subparagraph{Relation to previous work.}
The general bounds for arbitrary lines arise from complementary approaches. Aronov~\cite{Aronov2002} gives a quadratic lower-bound construction, while Sharir~\cite{Sharir1994} bounds the complexity of lower envelopes in higher dimensions. Agarwal, Aronov, and Sharir~\cite{AgarwalAronovSharir1997} develop algorithms for constructing such envelopes. Our result identifies a restricted family of lines for which the quadratic lower bound can be matched by an upper bound.

Papadopoulou and Wang~\cite{PapadopoulouWang2026} classify the nearest and farthest Voronoi diagrams of four lines in general position. In particular, they show that the number of vertices is even and that every value in $\{0,2,4,6,8\}$ is realizable. Their relations between the numbers of vertices, edges, faces, and cells~\cite[Theorem~9]{PapadopoulouWang2026} are used to convert our vertex bounds into bounds on total complexity. Our contribution is to restrict which four-site subsets can define vertices when the input contains arbitrarily many lines from one ruling.

Papadopoulou, Suderland, and Wang~\cite{PapadopoulouSuderlandWang2026} develop a collapse process for constructing farthest Voronoi diagrams of lines in $\IR^3$, starting from a spherical map of their unbounded features. They classify the events of this process and analyze local extrema of the distance function along trisectors. Their framework addresses general line inputs. For lines from a single ruling, we establish quadratic bounds on total complexity and obtain exact quadratic-time enumeration of all finite nearest and farthest vertices under the algebraic computation model stated in this paper.

Improved bounds are known for several restricted line families. Koltun and Sharir~\cite{KoltunSharir2003} prove a near-quadratic bound for lines with a fixed number of orientations. Weisman, Chew, and Kedem~\cite{WeismanChewKedem2004} obtain quadratic bounds for certain two-parameter families by linearizing squared-distance functions and analyzing lower envelopes of hyperplanes. Their family $\ell_{a,b}=\{(a,bt,t)\mid t\in\IR\}$ contains a ruling of the hyperbolic paraboloid $y=xz$ when $b=a$, although its lines are all parallel to a common plane and thus violate our general-position assumptions. Our proof uses a degree-four equation for tangency to a sphere and sign alternation to restrict the supporting lines to two pairs of consecutive sites in the cyclic order along the ruling, yielding quadratic bounds for both nearest and farthest diagrams. Their lower-bound construction does not enforce a single-ruling constraint, whereas ours realizes quadratically many nearest vertices within one ruling of a fixed surface. Glisse~\cite{Glisse2021} also proposes surfaces covered by straight lines as settings for lower-bound constructions.

The two-pair support condition is also related to the facet structure of cyclic polytopes. For a four-dimensional cyclic polytope with at least five vertices, Gale's evenness criterion describes each facet by two endpoint-disjoint pairs of consecutive vertices in cyclic order~\cite[Theorem~0.7]{Ziegler1995}. This is the combinatorial pattern obtained here from the signs of the contact quartic. The facet description alone does not establish that a prescribed four-line support admits a Euclidean sphere satisfying the required distance inequalities. Our lower-bound construction realizes quadratically many such supports as regular nearest vertices within one ruling of a fixed non-rotational hyperboloid, using a locally invertible map from sphere parameters to quartic coefficients.

Related representations of line families by conics appear in the work of Goodman et al.~\cite[Theorems~10, 12, and~13; Proposition~16]{GoodmanEtAl2006}. Their results provide geometric background for the coordinate representation used here. Our proof uses this representation to express Euclidean sphere tangency by a degree-four polynomial and then exploits its signs along the cyclic parameter order.

Algebraic descriptions have also been useful in other studies of Voronoi diagrams. Everett et al.~\cite{EverettEtAl2009} analyze three-line diagrams through the quadratic surfaces consisting of points equidistant from pairs of lines and the curves where these surfaces intersect. Cifuentes et al.~\cite{CifuentesEtAl2022} study Voronoi cells of sets defined by polynomial equations and the equations describing their boundaries. In our setting, the contact quartic connects these algebraic descriptions to both complexity bounds. Its signs restrict the possible vertex supports to $O(n^2)$ choices, while local control of its coefficients yields a matching lower-bound construction within a fixed ruling.


\section{Preliminaries}\label{sec:preliminaries}
We introduce the coordinates used to represent lines and establish the properties of these coordinates needed in our proofs. All distances and Voronoi diagrams are defined using the Euclidean metric in $\IR^3$. Standard references for the geometric background introduced below include Harris~\cite{Harris1992} and Pottmann and Wallner~\cite{PottmannWallner2001}.

\subsection{Euclidean distance and coordinates for lines}
Let
\[
\ell=a+\IR v = \{a+\lambda v \mid \lambda\in\IR\}, \qquad a \in \IR^3, \quad v \in \IR^3 \setminus\{0\},
\]
be an affine line. Here, $a$ is a point on the line and $v$ is a direction vector. For a point $c\in\IR^3$, orthogonal projection onto $\ell$ gives
\begin{equation}\label{eq:distance-projection}
    d(c,\ell)^2 = \|c-a\|^2 - \frac{\langle c-a,v\rangle^2}{\|v\|^2} = \frac{\|(c-a)\times v\|^2}{\|v\|^2}.
\end{equation}

To describe the line independently of the chosen point $a$, define $m=a\times v$. Replacing $a$ by another point $a+\lambda v$ on the line leaves $m$ unchanged. Replacing $v$ by $\mu v$, for $\mu\ne0$, multiplies both $v$ and $m$ by $\mu$. Thus, the six scalar coordinates $(v_1,v_2,v_3,m_1,m_2,m_3)$ represent the line up to multiplication by a common nonzero scalar. They satisfy $v \cdot m=0$. These coordinates determine the line. Indeed, the point on $\ell$ closest to the origin is
\[
    a_\perp=\frac{v\times m}{\|v\|^2},
\]
so that $\ell=a_\perp+\IR v$. Figure~\ref{fig:affine_line} illustrates these line coordinates and the orthogonal projections onto $\ell$. Since $c\times v-m=(c-a)\times v$, Equation~\eqref{eq:distance-projection} becomes
\begin{equation}\label{eq:plucker-distance}
    d(c,\ell)^2 = \frac{\|c\times v-m\|^2}{\|v\|^2}.
\end{equation}
Multiplying both $v$ and $m$ by the same nonzero scalar does not change this expression.

\begin{figure}
    \centering
    \includegraphics[width=4in]{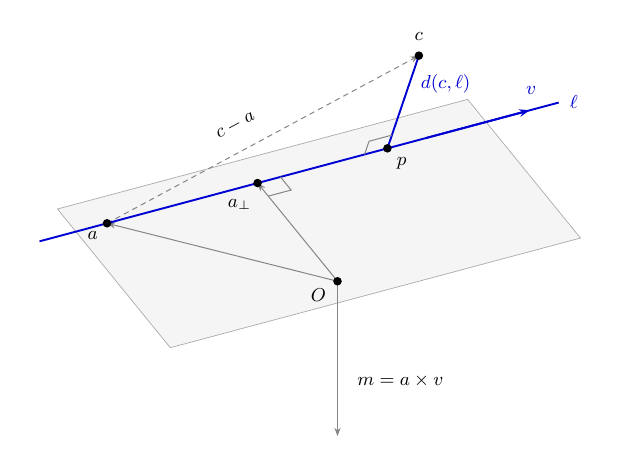}
    \caption{Euclidean distance and coordinates for a line $\ell=a+\IR v$. The shaded plane contains the origin $O$ and $\ell$, and the vector $m=a\times v$ is perpendicular to this plane. The points $a_\perp$ and $p$ are the orthogonal projections of $O$ and $c$ onto $\ell$, respectively, so that $d(c,\ell)=\|c-p\|$.}
    \label{fig:affine_line}
\end{figure}

\subsection{Projective coordinates and the Klein quadric}
The real projective space $\IP^d(\IR)$ consists of nonzero vectors in $\IR^{d+1}$, with two vectors representing the same point when one is a nonzero scalar multiple of the other. We write such a point as $[X_0:\cdots:X_d]$. The entries are called \emph{homogeneous coordinates}. A projective line consists of the points represented by the nonzero vectors in a two-dimensional linear subspace. A projective plane is defined in the same way using a three-dimensional linear subspace.

We regard Euclidean space as part of $\IP^3(\IR)$ through the map $(x_1,x_2,x_3)\longmapsto[1:x_1:x_2:x_3]$. The remaining points have $X_0=0$ and form the \emph{plane at infinity}. The projective line containing an affine line is called its \emph{projective closure}. A projective line that is not contained in the plane at infinity has an affine part that is an ordinary line in $\IR^3$. Figure~\ref{fig:projective-closure} illustrates how the projective closures of parallel affine lines meet at a common point at infinity.

\begin{figure}[tbp]
    \centering
    \includegraphics[width=\textwidth]{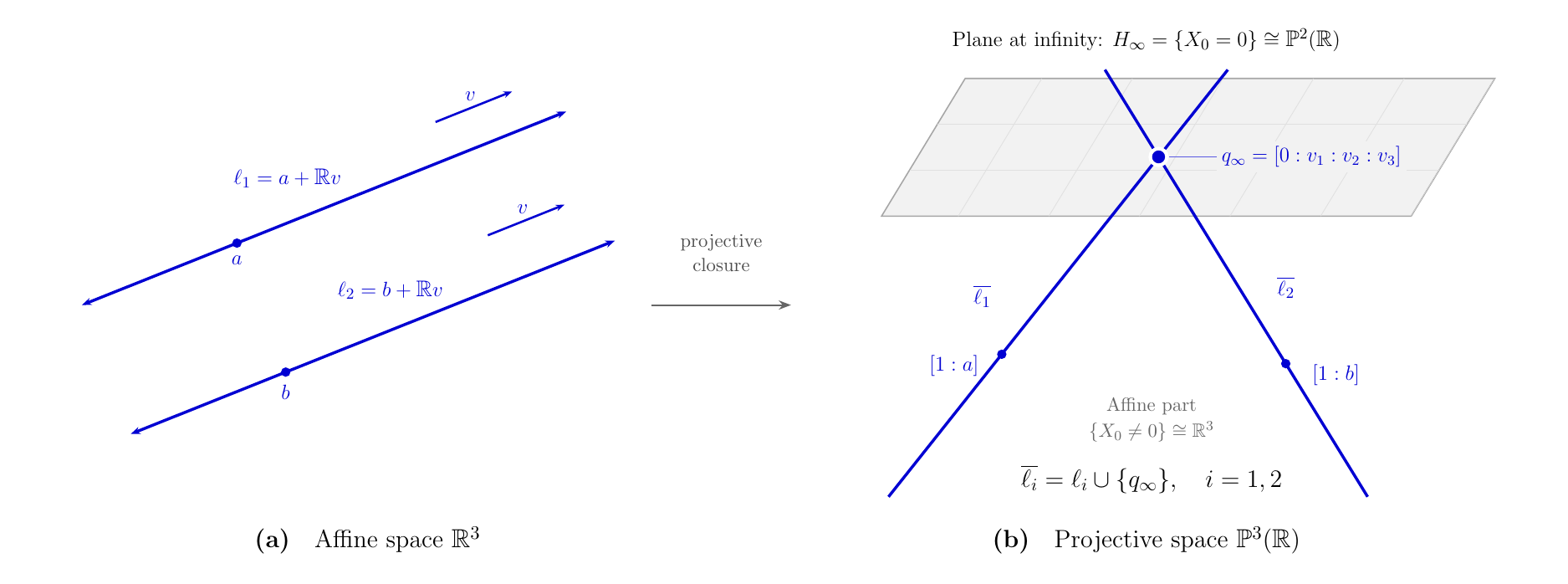}
    \caption{The plane at infinity and projective closure. (a) Two parallel affine lines $\ell_1=a+\IR v$ and $\ell_2=b+\IR v$ in $\IR^3$. (b) Their projective closures in $\IP^3(\IR)$ meet at the common point at infinity $q_\infty =[0:v_1:v_2:v_3]$ on $H_\infty=\{X_0=0\}$. Each closure is obtained by adding this single point to the corresponding affine line: $\overline{\ell_i}=\ell_i\cup\{q_\infty\}$.}
    \label{fig:projective-closure}
\end{figure}

A polynomial is \emph{homogeneous of degree $k$} if every monomial has total degree $k$. Multiplying all its arguments by $\lambda$ multiplies its value by $\lambda^k$. Consequently, its zero set is well defined in projective space.

Let a projective line be represented by the two-dimensional linear subspace spanned by independent vectors $x,y \in \IR^4$. Its \emph{Pl\"ucker coordinates} are
\[
    p_{ij}=x_i y_j-x_j y_i, \qquad 0\le i<j\le3.
\]
There are six such coordinates, one for each pair of indices. Changing the basis $x,y$ multiplies all six coordinates by a common nonzero scalar. They therefore define a point
\[
    [p_{01}:p_{02}:p_{03}:p_{12}:p_{13}:p_{23}] \in \IP^5(\IR).
\]
This representation of projective lines by projective points is called the \emph{Pl\"ucker embedding}.

The coordinates satisfy
\begin{equation}\label{eq:klein}
    p_{01}p_{23}-p_{02}p_{13}+p_{03}p_{12}=0.
\end{equation}
Conversely, every nonzero six-tuple satisfying this equation represents a unique projective line. The subset of $\IP^5(\IR)$ defined by Equation~\eqref{eq:klein} is called the \emph{Klein quadric}, and we denote it by $K$. Thus, the points of $K$ correspond bijectively to the projective lines in $\IP^3(\IR)$.

For an affine line $\ell=a+\IR v$, we may choose $x=(1,a_1,a_2,a_3)$ and $y=(0,v_1,v_2,v_3)$. The resulting coordinates satisfy $v=(p_{01},p_{02},p_{03})$ and $m=(p_{23},-p_{13},p_{12})$. Hence, the pair $(v,m)$ introduced above is precisely a signed reordering of the six Pl\"ucker coordinates. In this notation, Equation~\eqref{eq:klein} becomes $v\cdot m=0$, and we write the corresponding projective point as $[v:m]$. An affine line has $v\ne0$.
The points with $v=0$ represent projective lines contained in the plane at infinity, to which the Euclidean distance formula does not apply.

\subsection{Rulings and quadratic parametrizations}
A \emph{real projective quadric surface} is the zero set in $\IP^3(\IR)$ of a nonzero homogeneous quadratic polynomial $F(X_0,X_1,X_2,X_3)$. It is \emph{smooth} if the gradient of $F$ does not vanish at any nonzero vector representing a point of the surface.

A smooth quadric is \emph{doubly ruled over the reals} if it contains two one-parameter families of real projective lines, with exactly one line from each family passing through every real point of the surface. Each family is called a \emph{ruling}. Such a surface is also called a \emph{smooth split real quadric}.

After an invertible real projective change of coordinates, every such surface has the equation
\[
    Q_{\mathrm{Seg}} = \{[X_0:X_1:X_2:X_3]\in\IP^3(\IR) \mid X_0X_3-X_1X_2=0\}.
\]
Here, a projective change of coordinates means the transformation induced by an invertible linear map on the four homogeneous coordinates. We use such transformations to describe the line families. All distance calculations continue to use the original Euclidean coordinates.

Two distinct lines in the same ruling do not intersect in projective space. Their affine parts, when present, are therefore skew. Every line in one ruling intersects every line in the other ruling in projective space, although the intersection may lie at infinity. The one-sheeted hyperboloid and the hyperbolic paraboloid are standard affine examples. Their projective equations are obtained by introducing $X_0$ so that every term has degree two.

A \emph{plane conic} is a curve defined by a homogeneous quadratic equation within a projective plane. It is \emph{nonsingular} if the gradient of this equation, in the three coordinates of that plane, does not vanish at any point of the curve.

We use the following standard correspondence between rulings and plane conics; see Rudnev~\cite{Rudnev2018}.

\begin{proposition}\label{prop:ruling-conic}
Under the Pl\"ucker embedding, either ruling of a smooth doubly ruled real quadric is a nonsingular plane conic.
\end{proposition}

Fix one ruling and denote it by $\Gamma$. We also use $\Gamma$ for its image under the Pl\"ucker embedding, identifying a line with the point representing it when convenient. Let $\Pi_\Gamma$ denote the projective plane containing this conic. As verified in Appendix~\ref{app:recovering-conic}, $\Gamma=K\cap\Pi_\Gamma$. Figure~\ref{fig:ruling-conic} illustrates this correspondence.

\begin{figure}
    \centering
    \includegraphics[width=\linewidth]{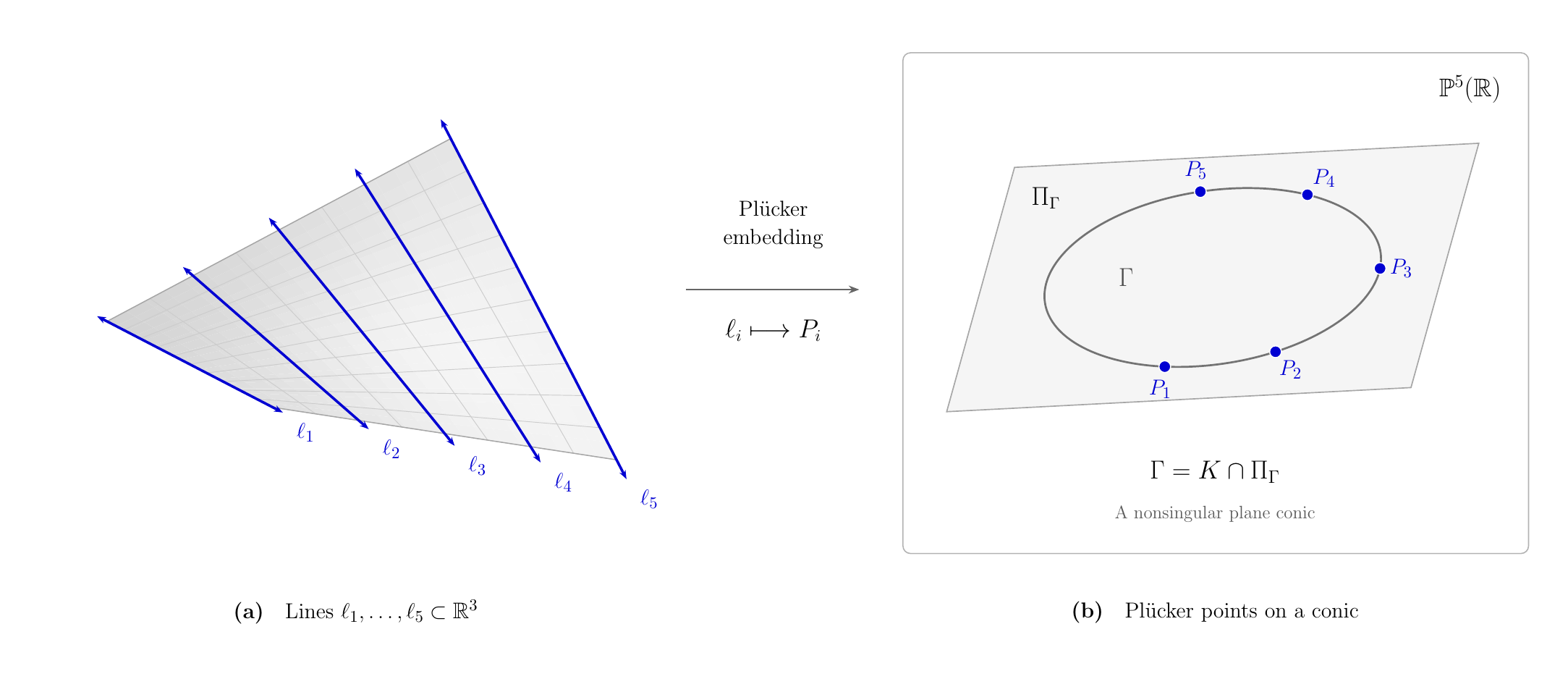}
    \caption{A ruling and its image under the Pl\"ucker embedding. (a) Five affine lines $\ell_1,\ldots,\ell_5$ from one ruling of the hyperbolic paraboloid $Q:x_3=x_1x_2$. (b) Their projective closures correspond to points $P_1,\ldots,P_5$ on the nonsingular plane conic $\Gamma=K\cap\Pi_\Gamma\subset\IP^5(\IR)$. The full projective ruling corresponds bijectively to $\Gamma$.}
    \label{fig:ruling-conic}
\end{figure}

We choose a one-to-one parametrization $\gamma:\mathbb P^1(\mathbb R)\longrightarrow\Gamma\subset K$ of the form
\[
    \gamma([s:t]) = [g_{01}(s,t):g_{02}(s,t):g_{03}(s,t): g_{12}(s,t):g_{13}(s,t):g_{23}(s,t)],
\]
where each $g_{ij}$ is a homogeneous polynomial of degree two. The six polynomials have no common nonconstant factor and do not vanish simultaneously at any parameter $[s:t]$. We call a parametrization with no common nonconstant factor \emph{reduced}.

Using the signed reordering described above, we write 
\[
v(s,t)=(g_{01},g_{02},g_{03}) \quad \text{and} \quad m(s,t)=(g_{23},-g_{13},g_{12}),
\]
where each polynomial is evaluated at $(s,t)$. Thus, all six entries of $v(s,t)$ and $m(s,t)$ are homogeneous quadratic polynomials. A parameter represents an affine line precisely when $v(s,t)\ne0$.

Every point of $\IP^1(\IR)$ can be written as $[\cos\theta:\sin\theta]$, with $\theta$ taken modulo $\pi$. The parameter space is therefore topologically a circle and carries a cyclic order. We refer to it as the \emph{parameter circle}. On the part where $s\ne0$, we may instead use the ordinary real parameter $u=t/s$. The single remaining parameter is $[0:1]$.

Related uses of Pl\"ucker conics to describe ruled surfaces appear in Goodman et al.~\cite{GoodmanEtAl2006}.

\subsection{General position and combinatorial results}\label{sec:general-position}
Let $L=\{\ell_1,\ldots,\ell_n\}$ be $n\ge4$ distinct affine lines from the chosen ruling $\Gamma$. A sphere with center $c$ and radius $r>0$ is tangent to a line $\ell$ when $d(c,\ell)=r$. The point of tangency is the orthogonal projection of $c$ onto $\ell$.

For a line with direction vector $v$, its two unit direction vectors are $v/\|v\|$ and $-v/\|v\|$. They are points of the unit sphere $S^2=\{x\in\IR^3\mid \|x\|=1\}$, which we call the \emph{sphere of directions}. Points on this sphere are \emph{cocircular} if they lie on a common circle obtained by intersecting the sphere with a plane.

We use the general-position conditions of Papadopoulou and Wang~\cite{PapadopoulouWang2026}, together with one additional condition on the chosen ruling.

\begin{assumption}\label{ass:gp}
The set $L$ and its ruling $\Gamma$ satisfy the following conditions.
\begin{enumerate}
    \item The lines are pairwise skew, and no three are parallel to a common plane.
    \item No sphere is tangent to four distinct sites at four coplanar contact points.
    \item No sphere is tangent to five distinct sites.
    \item For any four distinct sites, no choice of one unit direction vector for each site gives four cocircular points on $S^2$.
    \item No Euclidean sphere is tangent to every affine line in the chosen ruling $\Gamma$.
\end{enumerate}
\end{assumption}

The first four conditions concern the input lines. Pairwise skewness follows automatically from membership in one ruling. The fifth condition restricts the ruling itself by excluding a sphere tangent to all its affine lines. In the present setting, condition~(3) follows from condition~(5). Indeed, by Theorem~\ref{thm:degree-four}, a sphere tangent to five distinct sites would have an identically zero contact polynomial and would therefore be tangent to every affine line of the ruling. We retain condition~(3) to match the general-position assumptions of Papadopoulou and Wang~\cite{PapadopoulouWang2026}. Section~\ref{sec:lower-bound} shows that the inputs used in our lower-bound construction can satisfy all five conditions.

For later use, we also specify what is meant by a regular vertex. A site \emph{supports} a nearest or farthest Voronoi vertex if it attains the minimum or maximum distance, respectively, at that vertex. A vertex $c$ is \emph{regular} if exactly four sites support it and, after labeling these sites $\ell_1,\ldots,\ell_4$, the three vectors
\[
    \left. \nabla_x \bigl(d(x,\ell_i)^2-d(x,\ell_1)^2\bigr) \right|_{x=c}, \qquad i=2,3,4,
\]
are linearly independent. In geometric terms, the three surfaces expressing equality of these distances meet with independent normal directions.

This independence has a direct interpretation in terms of the contact points. If $q_i$ is the orthogonal projection of $c$ onto $\ell_i$, then
\[
    \left. \nabla_x \bigl(d(x,\ell_i)^2-d(x,\ell_1)^2\bigr) \right|_{x=c} = 2(q_1-q_i).
\]
Thus, the three gradients are linearly independent exactly when the four contact points are not coplanar.

Finally, let $V_N$, $E_N$, $F_N$, and $C_N$ denote the numbers of vertices, edges, two-dimensional faces, and three-dimensional cells of the nearest Voronoi diagram. Define $V_F$, $E_F$, $F_F$, and $C_F$ analogously for the farthest Voronoi diagram. Three-dimensional cells are counted as connected regions. In particular, the set of points having a given line as their unique farthest site may have several connected components.

Under Assumption~\ref{ass:gp}, Papadopoulou and Wang~\cite[Theorem~9]{PapadopoulouWang2026} give
\begin{align}
    E_N &= 2V_N+2n-2, & F_N &= V_N+3n-3, & C_N &= n, \label{eq:nvd-features}\\
    E_F &= 2V_F+n^2-n-2, & F_F &= V_F+2n^2-2n-3, & C_F &= n^2-n. \label{eq:fvd-features}
\end{align}

The total combinatorial complexity, denoted by $\operatorname{comp}$, is the sum of the numbers of features in all four dimensions. Writing $\NVD(L)$ and $\FVD(L)$ for the nearest and farthest Voronoi diagrams, respectively, we obtain
\begin{align}
    \operatorname{comp}(\NVD(L)) &=4V_N+6n-5, \label{eq:nukvd-complexity}\\
    \operatorname{comp}(\FVD(L)) &=4V_F+4n^2-4n-5. \label{eq:fukvd-complexity}
\end{align}
These identities allow us to bound the total complexity of each diagram by bounding its number of vertices.


\section{The Contact Quartic and Quadratic Upper Bounds}\label{sec:contact-and-upper-bounds}
Let $L=\{\ell_1,\ldots,\ell_n\}$ be a set of $n\ge 4$ affine line sites in the chosen ruling, satisfying Assumption~\ref{ass:gp}. We first express tangency to a sphere by a polynomial of degree four in the ruling parameter. When a sphere is tangent to four sites, these four contacts determine all roots of the polynomial. Its sign then restricts the supporting sites of a Voronoi vertex to the endpoints of two gaps in their cyclic order. There are only quadratically many such gap pairs, and each pair supports at most eight vertices.

\subsection{The contact quartic}
By Proposition~\ref{prop:ruling-conic}, the chosen ruling is represented by a nonsingular plane conic. We use a one-to-one parametrization $\gamma:\mathbb{P}^{1}(\mathbb{R})\longrightarrow\Gamma\subset K$ of the form
\begin{equation}
\label{eq:ruling-parametrization}
    \gamma([s:t]) = [p_{01}(s,t):p_{02}(s,t):p_{03}(s,t):p_{12}(s,t):p_{13}(s,t):p_{23}(s,t)],
\end{equation}
where the six coordinate polynomials are homogeneous of degree two and have no common factor. Thus, every monomial in each coordinate polynomial has total degree two.

In the vector--moment coordinates introduced in Section~\ref{sec:preliminaries}, write
\[
    v(s,t) = \bigl(p_{01}(s,t),p_{02}(s,t),p_{03}(s,t)\bigr), \qquad m(s,t) = \bigl(p_{23}(s,t),-p_{13}(s,t),p_{12}(s,t)\bigr).
\]
Both vectors therefore have quadratic polynomial coordinates. For a parameter $\tau=[s:t]$ representing an affine line, denote that line by $\ell_\tau$. In particular, $\tau$ is a point on the parameter circle, whereas $\ell_\tau$ is a line in Euclidean space.

Consider a sphere with center $c\in\mathbb{R}^{3}$ and radius $r>0$, and write $R=r^2$. Define
\[
    \Phi_{c,R}(v,m) = \|c\times v-m\|^{2}-R\|v\|^{2}.
\]
The distance formula~\eqref{eq:plucker-distance} gives
\begin{equation}\label{eq:contact-distance-relation}
    \Phi_{c,R}(v,m) = \|v\|^{2}\bigl(d(c,\ell)^{2}-R\bigr)
\end{equation}
for every affine line $\ell$ represented by $(v,m)$. Since $v\ne0$ for an affine line, the sign of this expression records whether its distance from $c$ is smaller than, equal to, or larger than the sphere radius.

\begin{definition}\label{def:contact-quartic}
The \emph{contact polynomial} of the sphere with center $c$ and squared radius $R$ is
\begin{equation}\label{eq:contact-quartic}
    H_{c,R}(s,t) = \Phi_{c,R}\bigl(v(s,t),m(s,t)\bigr).
\end{equation}
\end{definition}

\begin{theorem}\label{thm:degree-four}
The contact polynomial $H_{c,R}$ is either identically zero or homogeneous of degree four. A homogeneous polynomial of degree four in two variables is called a \emph{binary quartic}. At every parameter $\tau=[s:t]$ representing an affine line,
\begin{equation}\label{eq:contact-sign}
    \operatorname{sgn}H_{c,R}(s,t) = \operatorname{sgn}\bigl(d(c,\ell_\tau)^2-R\bigr).
\end{equation}
Its sign is independent of the representative $(s,t)$ of $\tau$.

If $H_{c,R}$ is not identically zero and the sphere is tangent to four distinct affine lines of the ruling, their parameters are its only roots on the parameter circle. All four roots are simple, and the sign alternates on the four open arcs between consecutive roots.
\end{theorem}

\begin{proof}
Every coordinate of $c\times v(s,t)-m(s,t)$ is homogeneous of degree two. Taking its squared norm produces a homogeneous polynomial of degree four, as does $R\|v(s,t)\|^2$. Their difference is therefore either zero or homogeneous of degree four.

Equation~\eqref{eq:contact-distance-relation} proves the sign identity. Moreover, replacing $(s,t)$ by $(\lambda s,\lambda t)$, where $\lambda\ne0$, multiplies $H_{c,R}$ by $\lambda^4>0$. Thus, the sign is well-defined on $\mathbb{P}^{1}(\mathbb{R})$.

Each tangent line gives a root of $H_{c,R}$. A nonzero homogeneous polynomial of degree four has at most four roots on the projective parameter line, counted with multiplicity. Consequently, four distinct contact parameters account for all its roots, and each root has multiplicity one. The sign changes when passing through a simple real root. It therefore alternates on the four intervening arcs.
\end{proof}

We call a nonzero $H_{c,R}$ the \emph{contact quartic}. For brevity, an inequality such as $H_{c,R}(\tau)<0$ means
that the inequality holds for any nonzero representative of $\tau$.

Let $B(c,r)$ be the open Euclidean ball of radius $r$ centered at $c$, let $\overline{B}(c,r)$ be its closure, and let $\partial B(c,r)$ be its boundary sphere. For every parameter $\tau=[s:t]$ representing an affine line,
\begin{align}
    H_{c,R}(s,t)<0 &\iff \ell_\tau\text{ meets }B(c,\sqrt R), \label{eq:negative-contact}\\
    H_{c,R}(s,t)=0 &\iff \ell_\tau\text{ is tangent to }\partial B(c,\sqrt R),\\
    H_{c,R}(s,t)>0 &\iff \ell_\tau\text{ misses }\overline{B}(c,\sqrt R). \label{eq:positive-contact}
\end{align}
These interpretations apply to affine lines. The polynomial itself remains defined on the entire projective parameter circle.

\subsection{Cyclic supports}
At a nearest Voronoi vertex $c$, the \emph{supporting sites} are the lines attaining the minimum distance from $c$. For a farthest vertex, they are the lines attaining the maximum. Their common distance is the \emph{supporting radius}. It is positive because the sites are pairwise skew.

Under Assumption~\ref{ass:gp}, every finite Voronoi vertex is regular: it has exactly four supporting sites, and their contact points are not coplanar. Choose an orientation of the parameter circle and list the site parameters in cyclic order as $\tau_1,\ldots,\tau_n$. Indices are taken modulo $n$, so $\tau_{n+1}=\tau_1$.

\begin{definition}\label{def:cyclic-gaps}
The \emph{cyclic gap} $g_i$ is the open arc from $\tau_i$ to $\tau_{i+1}$ in the chosen orientation. It contains no site parameter. Two gaps are \emph{endpoint-disjoint} if their four endpoints are distinct.
\end{definition}

\begin{lemma}\label{lem:gap-pair-count}
The number of unordered pairs of endpoint-disjoint cyclic gaps is
\[
    \binom{n}{2}-n=\frac{n(n-3)}{2}.
\]
\end{lemma}

\begin{proof}
There are $\binom{n}{2}$ unordered pairs of distinct gaps. Exactly $n$ pairs share an endpoint, one pair at each site
parameter. Removing these pairs gives the claimed count.
\end{proof}

\begin{lemma}\label{lem:cyclic-support}
Let $c$ be a regular Voronoi vertex, and let $R$ be the square of its supporting radius. Then the following statements hold.
\begin{enumerate}
    \item
    If $c$ is a nearest vertex, the two arcs on which $H_{c,R}$ is negative are endpoint-disjoint cyclic gaps. Their four endpoints are the supporting site parameters.
    \item
    If $c$ is a farthest vertex, the two arcs on which $H_{c,R}$ is positive are endpoint-disjoint cyclic gaps. Their four endpoints are the supporting site parameters.
\end{enumerate}
\end{lemma}

\begin{proof}
The supporting sphere is tangent to four distinct sites. Assumption~\ref{ass:gp}(5) excludes a sphere tangent to every affine line of the ruling, so its contact polynomial is not identically zero. By Theorem~\ref{thm:degree-four}, the four supporting parameters are simple roots, and the sign alternates between them. There are therefore two negative arcs and two positive arcs.

For a nearest vertex, every site has distance at least $\sqrt R$ from $c$. Equation~\eqref{eq:contact-sign} gives
\[
    H_{c,R}(\tau_i)\ge0 \qquad\text{for every site parameter }\tau_i.
\]
Neither negative arc can contain a site parameter. Each is therefore a gap between consecutive sites. The two negative arcs have four distinct endpoints because they alternate with the two positive arcs.

For a farthest vertex, every site has distance at most $\sqrt R$ from $c$, so
\[
    H_{c,R}(\tau_i)\le0 \qquad\text{for every site parameter }\tau_i.
\]
Neither positive arc can contain a site parameter. The same argument shows that the positive arcs are two endpoint-disjoint cyclic gaps.
\end{proof}

Figure~\ref{fig:cyclic-supports} illustrates the two sign patterns and the corresponding endpoint-disjoint gap pairs.

\begin{figure}
    \centering
    \includegraphics[width=\linewidth]{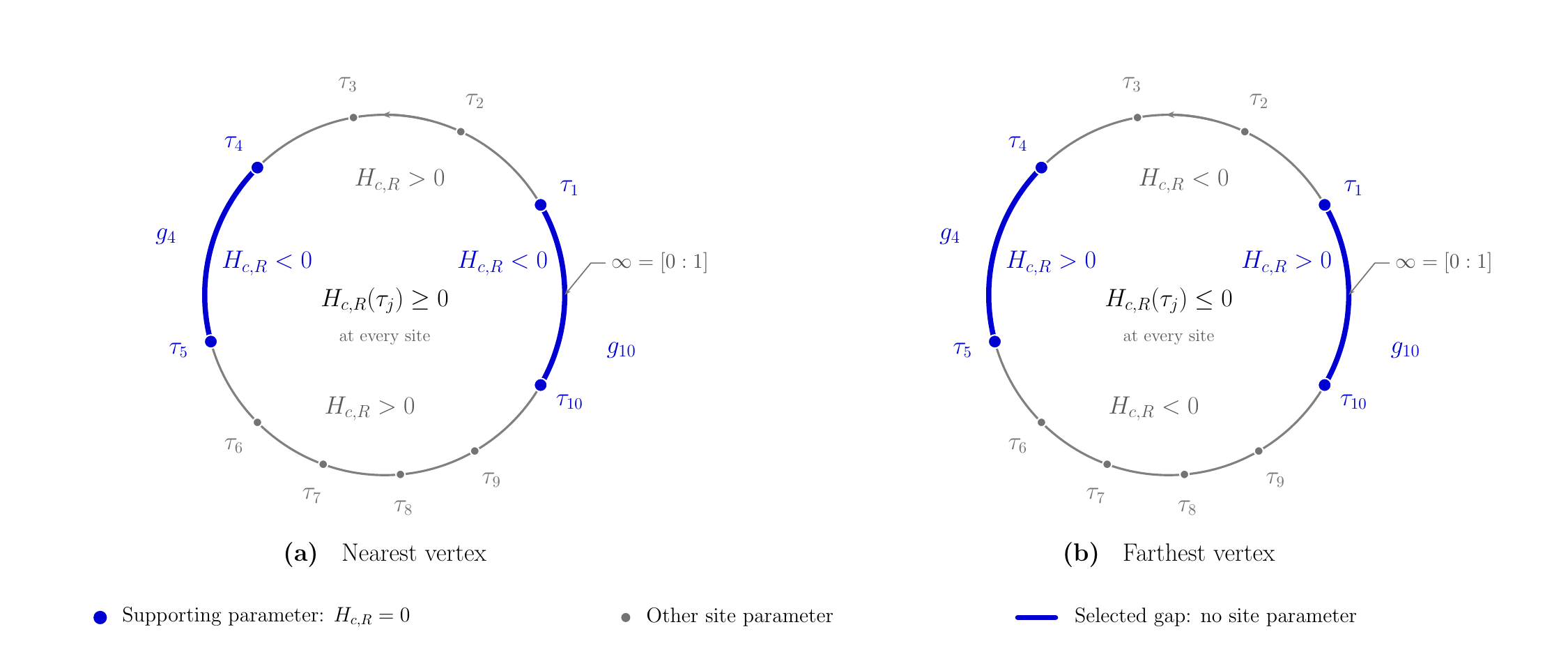}
    \caption{Cyclic supports on the parameter circle. Each marked point represents an input line. At a regular Voronoi vertex, the four blue site parameters are simple zeros of the contact quartic $H_{c,R}$. The highlighted arcs $g_4$ and $g_{10}$ contain no site parameters and form an endpoint-disjoint pair of cyclic gaps: (a) the two negative arcs for a nearest vertex, (b) the two positive arcs for a farthest vertex. The gap $g_{10}$ passes through $\infty=[0:1]$, the parameter omitted by the affine coordinate $u=t/s$, illustrating the need to retain the cyclic order.}
    \label{fig:cyclic-supports}
\end{figure}

\begin{remark}
The parameter space is a circle, not a line. Using a single real coordinate leaves out one parameter value, and a gap may pass through that omitted value $\infty = [0 : 1]$, as illustrated by $g_{10}$ in Figure~\ref{fig:cyclic-supports}. Keeping the cyclic order ensures that this gap is also counted.
\end{remark}

\subsection{Quadratic upper bounds}
The preceding argument limits the possible supporting sites to the endpoints of quadratically many gap pairs. We now bound the number of vertices that four fixed sites can support.

\begin{lemma}\label{lem:eight-center-bound}
A fixed quadruple of line sites supports at most eight regular real Voronoi vertices.
\end{lemma}

\begin{proof}
Represent the four lines by vector--moment coordinates $(v_i,m_i)$, for $i=1,2,3,4$, and define
\[
    q_i(x)=\|x\times v_i-m_i\|^2, \qquad w_i=\|v_i\|^2>0.
\]
A point $x\in\mathbb{R}^3$ is equidistant from the four lines if and only if
\[
    \frac{q_i(x)}{w_i} - \frac{q_1(x)}{w_1} = 0, \qquad i=2,3,4.
\]
Multiplying by the positive constant $w_1w_i$ gives
\[
    w_1q_i(x)-w_iq_1(x)=0, \qquad i=2,3,4.
\]
These are three polynomial equations of degree at most two in the three coordinates of $x$.

At a regular vertex, the three squared-distance differences have linearly independent gradients. Multiplication by the nonzero constants $w_1w_i$ preserves this independence. Thus, the vertex is an isolated nonsingular solution, also when the variables are allowed to be complex. B\'ezout's inequality bounds the number of isolated complex solutions by the product of the three degrees, which is at most $2^3=8$. This also bounds the number of regular real vertices; see Harris~\cite{Harris1992} and Koltun and Sharir~\cite{KoltunSharir2003}.
\end{proof}

Let $V_N$ and $V_F$ denote the numbers of vertices of the nearest and farthest diagrams, respectively. As in Section~\ref{sec:general-position}, $\operatorname{comp}$ counts all vertices, edges, two-dimensional faces, and three-dimensional cells.

\begin{theorem}\label{thm:quadratic-upper-bounds}
Under Assumption~\ref{ass:gp},
\[
    V_N\le4n(n-3), \qquad V_F\le4n(n-3).
\]
Consequently,
\begin{align*}
    \operatorname{comp}(\NVD(L)) &\le16n^2-42n-5=O(n^2),\\
    \operatorname{comp}(\FVD(L)) &\le20n^2-52n-5=O(n^2).
\end{align*}
Moreover, every such farthest diagram has total combinatorial complexity $\Theta(n^2)$.
\end{theorem}

\begin{proof}
Associate each nearest vertex with the pair of negative gaps of its contact quartic, and each farthest vertex with the pair of positive gaps. Lemma~\ref{lem:cyclic-support} makes this association well-defined.

By Lemma~\ref{lem:gap-pair-count}, there are $n(n-3)/2$ possible gap pairs. Each pair determines four endpoint lines.
Lemma~\ref{lem:eight-center-bound} gives at most eight regular vertices supported by those lines. Therefore,
\[
    V_N,V_F \le 8\cdot\frac{n(n-3)}2 = 4n(n-3).
\]

The feature-count formulas of Papadopoulou and Wang, recalled in Equations~\eqref{eq:nvd-features} and~\eqref{eq:fvd-features}, give
\[
    \operatorname{comp}(\NVD(L))=4V_N+6n-5 \qquad \text{and} \qquad \operatorname{comp}(\FVD(L))=4V_F+4n^2-4n-5.
\]
Substituting the vertex bounds proves the stated upper bounds.

Finally, every farthest diagram has $n^2-n$ three-dimensional cells. Its total complexity is therefore also $\Omega(n^2)$, which proves the $\Theta(n^2)$ bound for every input satisfying our assumptions.
\end{proof}

\begin{remark}
For the farthest diagram, the quadratic lower bound already follows from the number of three-dimensional cells. For the nearest diagram, a matching worst-case lower bound requires the separate construction in Section~\ref{sec:lower-bound}.
\end{remark}


\section{A Lower Bound on a Fixed Ruling}\label{sec:lower-bound}
We construct a set of lines whose nearest Voronoi diagram has quadratically many regular vertices. All lines belong to one ruling of the fixed non-rotational one-sheeted hyperboloid
\begin{equation}\label{eq:fixed-hyperboloid}
    Q_*:\frac{x^2}{4}+y^2-z^2=1.
\end{equation}
The construction starts with a sphere whose contact polynomial has a root of multiplicity four. By moving its center and changing its radius, we realize nearby polynomials with four distinct roots. We choose these roots to be the endpoints of two gaps between consecutive sites. The sign of the resulting polynomial then guarantees that the supporting ball is empty.

One ruling of $Q_*$ consists of the lines
\[
    \ell_\theta = \left\{ \bigl( 2(\cos\theta-\lambda\sin\theta), \sin\theta+\lambda\cos\theta, \lambda \bigr) \;\middle|\; \lambda\in\mathbb{R} \right\}.
\]
Here $\theta$ selects a line, while $\lambda$ specifies a point on that line. We use lines with $\theta$ close to zero and introduce the parameter $u=\tan(\theta/2)$. Write $\ell_u$ for the line corresponding to $\theta=2\arctan u$.

Multiplying a direction vector and its moment by the common positive factor $1+u^2$ gives
\[
    V(u)=(-4u,\,1-u^2,\,1+u^2) \qquad \text{and} \qquad M(u)=(2u,\,-2(1-u^2),\,2(1+u^2)).
\]
This common scaling does not change the represented line. The vectors satisfy $V(u)\cdot M(u)=0$, and $V(u)$ is nonzero for every real $u$.

For a sphere with center $c=(X,Y,Z)$ and squared radius $R>0$, its contact polynomial in this parameter is
\begin{equation}\label{eq:lower-bound-contact-polynomial}
    P_{c,R}(u) = \|c\times V(u)-M(u)\|^2-R\|V(u)\|^2 = \sum_{k=0}^{4}p_k(c,R)u^k.
\end{equation}
Thus, $p_k(c,R)$ denotes the coefficient of $u^k$. By Equation~\eqref{eq:contact-distance-relation}, $P_{c,R}(u) = \|V(u)\|^2\bigl(d(c,\ell_u)^2-R\bigr)$. Its roots therefore identify tangent lines, and its sign determines whether a line meets the open ball or misses the closed ball.

\subsection{A contact polynomial with a root of multiplicity four}
\begin{lemma}\label{lem:fourth-order-contact}
For
\[
    c_*=\left(\frac52,0,0\right), \qquad R_*=\frac14,
\]
the contact polynomial is
\[
    P_{c_*,R_*}(u)=40u^4.
\]
\end{lemma}

\begin{proof}
Substituting $c_*$ gives 
\[
c_*\times V(u)-M(u) = \left( -2u, -\frac{1+9u^2}{2}, \frac{1-9u^2}{2} \right) \qquad \text{and} \qquad \|V(u)\|^2=2+16u^2+2u^4.
\]
Taking the squared norm of the first expression and subtracting $R_*\|V(u)\|^2$ yields $40u^4$. The full coefficient expansion for arbitrary $(c,R)$ is given in Appendix~\ref{app:contact-expansion}.
\end{proof}

The parameter $u=0$ is therefore a root of multiplicity four of the contact polynomial. This multiplicity concerns variation of the line within the ruling. The corresponding line $\ell_0$ is tangent to the sphere. We use this sphere as a starting point and vary its four parameters to obtain four distinct nearby contact parameters.

\subsection{Local realization of contact quartics}
At $(c_*,R_*)$, the leading coefficient is $p_4=40$. It remains positive in a sufficiently small neighborhood. We may therefore divide the contact polynomial by $p_4$ to make its leading coefficient equal to one. A polynomial with leading coefficient one is called \emph{monic}.

After this normalization, four coefficients remain. Define the map
\[
    \Psi(X,Y,Z,R) = 
    \left(
        \frac{p_0}{p_4},
        \frac{p_1}{p_4},
        \frac{p_2}{p_4},
        \frac{p_3}{p_4}
    \right).
\]
Thus, $\Psi$ maps the three center coordinates and the squared radius to the four remaining coefficients. Lemma~\ref{lem:fourth-order-contact} gives $\Psi(c_*,R_*)=(0,0,0,0)$.

Let $D\Psi$ denote the matrix of first partial derivatives of $\Psi$ with respect to $X,Y,Z,R$. Its invertibility will allow us to prescribe small changes in the four polynomial coefficients.

\begin{lemma}\label{lem:nonsingular-coefficient-map}
At the sphere $(c_*,R_*)$,
\[
    \det D\Psi(c_*,R_*)=\frac1{250}\ne0.
\]
\end{lemma}

\begin{proof}
Differentiate the coefficients in Equation~\eqref{eq:lower-bound-contact-polynomial} and evaluate at $(c_*,R_*)$. The resulting matrix is
\[
    D\Psi(c_*,R_*) =
    \begin{pmatrix}
        \frac1{20} & 0  & 0          & -\frac1{20}\\
        0          & 0  & \frac15    & 0\\
        0          & 0  & 0          & -\frac25\\
        0          & -1 & \frac45    & 0
    \end{pmatrix},
\]
whose determinant is $1/250$.
The coefficient calculations are detailed in Appendix~\ref{app:coefficient-jacobian}.
\end{proof}

\begin{corollary}\label{cor:local-quartic-realization}
There are neighborhoods $\mathcal W$ of $(c_*,R_*)$ and $\mathcal U$ of the origin in $\mathbb{R}^4$ such that every monic quartic
\[
    u^4+a_3u^3+a_2u^2+a_1u+a_0, \qquad (a_0,a_1,a_2,a_3)\in\mathcal U,
\]
is the normalized contact polynomial of a unique sphere with parameters $(c,R)\in\mathcal W$. Throughout $\mathcal W$, the squared radius $R$ and the leading coefficient $p_4$ are positive, and $D\Psi$ is invertible.
\end{corollary}

\begin{proof}
The real inverse function theorem applies to $\Psi$ at $(c_*,R_*)$ by Lemma~\ref{lem:nonsingular-coefficient-map}. It provides neighborhoods on which $\Psi$ has an inverse. By continuity, we may choose these neighborhoods small enough that $R>0$, $p_4>0$, and $\det D\Psi\ne0$ throughout $\mathcal W$.
\end{proof}

Consequently, any sufficiently small change in the four coefficients of $u^4$ can be achieved by a unique change of the sphere within $\mathcal W$. In particular, we can prescribe four distinct roots sufficiently close to zero.

\subsection{The adjacent-gap construction}
Choose distinct parameters
\[
    u_1<u_2<\cdots<u_n
\]
in a sufficiently small interval about zero, and use $\ell_{u_1},\ldots,\ell_{u_n}$ as the input sites.
For $1\le i\le n-3$ and $i+2\le j\le n-1$ define
\begin{equation}\label{eq:adjacent-gap-quartic}
    Q_{ij}(u) = (u-u_i)(u-u_{i+1})(u-u_j)(u-u_{j+1}).
\end{equation}
The four roots are the endpoints of two gaps between consecutive sites, and the two gaps have no common endpoint.

As the roots approach zero, the four coefficients below the leading coefficient of $Q_{ij}$ approach zero. We can therefore choose the interval small enough that the coefficient vector of every $Q_{ij}$ belongs to $\mathcal U$. Corollary~\ref{cor:local-quartic-realization} then supplies a sphere $S_{ij}$ with center $c_{ij}$ and squared radius $R_{ij}$ such that
\[
    P_{c_{ij},R_{ij}}(u) = p_4(c_{ij},R_{ij})Q_{ij}(u), \qquad p_4(c_{ij},R_{ij})>0.
\]

\begin{lemma}\label{lem:empty-adjacent-gap-spheres}
The sphere $S_{ij}$ is tangent to precisely the four input lines
\[
    \ell_{u_i},\quad \ell_{u_{i+1}},\quad \ell_{u_j},\quad \ell_{u_{j+1}}.
\]
Its open ball meets no input line. Every other input line misses its closed ball.
\end{lemma}

\begin{proof}
The four roots of $Q_{ij}$ give the four tangent lines. Since its leading coefficient is positive and its roots are distinct and ordered, $Q_{ij}$ is negative precisely on $(u_i,u_{i+1})$ and $(u_j,u_{j+1})$. Neither interval contains an input parameter. At every other input parameter, $Q_{ij}$ is strictly positive.

The actual contact polynomial is a positive multiple of $Q_{ij}$ and has the same sign. Equation~\eqref{eq:positive-contact} therefore shows that every non-supporting input line misses the closed ball. The four supporting lines are tangent to its boundary, so they also miss its interior.
\end{proof}

\subsection{Regularity and general position}
We first show that the constructed spheres give regular vertices. Fix one sphere $S_{ij}$ and denote its four contact parameters, in increasing order, by $r_1,r_2,r_3,r_4$. For these fixed parameters, define the four functions
\[
    G_k(c,R)=P_{c,R}(r_k), \qquad k=1,2,3,4,
\]
and let $G=(G_1,G_2,G_3,G_4)$. The four tangency equations are $G(c,R)=0$.

Write $(a_0,a_1,a_2,a_3)=\Psi(c,R)$. At the constructed sphere, differentiation of
\[
    G_k(c,R) = p_4(c,R) \bigl(r_k^4+a_3r_k^3+a_2r_k^2+a_1r_k+a_0\bigr)
\]
gives
\[
    DG = p_4
    \begin{pmatrix}
        1&r_1&r_1^2&r_1^3\\
        1&r_2&r_2^2&r_2^3\\
        1&r_3&r_3^2&r_3^3\\
        1&r_4&r_4^2&r_4^3
    \end{pmatrix}
    D\Psi.
\]
Here $DG$ is the matrix of first derivatives with respect to the three center coordinates and the squared radius. The terms involving derivatives of $p_4$ vanish because each $r_k$ is a root of the normalized contact polynomial.

The middle matrix is a Vandermonde matrix with distinct parameters $r_1,r_2,r_3,r_4$, so its determinant is $\prod_{1\le a<b\le4}(r_b-r_a)\ne0$. Since $p_4>0$ and $D\Psi$ is invertible throughout $W$, the matrix $DG$ is also invertible.

By the distance formula, each equation $G_k(c,R)=0$ is a positive constant multiple of $d(c,\ell_{r_k})^2-R=0$. The derivative matrix of these four distance equations is therefore invertible as well. Subtracting the first equation from the other three eliminates $R$ and shows that the gradients of
\[
    d(c,\ell_{r_k})^2-d(c,\ell_{r_1})^2, \qquad k=2,3,4,
\]
are linearly independent. Together with Lemma~\ref{lem:empty-adjacent-gap-spheres}, this proves that $c_{ij}$ is a regular nearest Voronoi vertex.

It remains to choose the parameters so that the entire input satisfies Assumption~\ref{ass:gp}. These assumptions also concern spheres and subsets of sites beyond those used in the construction.

\begin{lemma}\label{lem:lower-bound-genericity}
For every $n\ge4$, the parameters $u_1<\cdots<u_n$ can be chosen in an arbitrarily small interval about zero so that Assumption~\ref{ass:gp} holds and every sphere $S_{ij}$ constructed from Equation~\eqref{eq:adjacent-gap-quartic} determines a regular nearest Voronoi vertex.
\end{lemma}

\begin{proof}
The parameter-choice argument establishing all conditions of Assumption~\ref{ass:gp} is given in Appendix~\ref{app:lower-bound-details}. Choose the parameters there within an interval small enough for Corollary~\ref{cor:local-quartic-realization} to apply to every $Q_{ij}$. The preceding derivative calculation proves regularity of all the constructed vertices, and Lemma~\ref{lem:empty-adjacent-gap-spheres} proves the required emptiness of their supporting balls.
\end{proof}

\begin{theorem}\label{thm:quadratic-nearest-lower}
For every $n\ge4$, one ruling of the fixed surface 
\[
    Q_*:\frac{x^2}{4}+y^2-z^2=1
\]
contains a set of $n$ affine lines satisfying Assumption~\ref{ass:gp} whose nearest Voronoi diagram has at least $\frac{(n-2)(n-3)}2$ distinct regular vertices.
\end{theorem}

\begin{proof}
Choose the parameters as in Lemma~\ref{lem:lower-bound-genericity}. Every index pair $1\le i\le n-3$ and $i+2\le j\le n-1$ gives a regular nearest vertex at the center of $S_{ij}$.

Different index pairs give different sets of four roots and hence different normalized contact polynomials. They therefore give different spheres. Moreover, two empty supporting spheres with the same center would have the same radius, equal to the minimum distance from that center to the input lines. They would then be the same sphere. Thus, the constructed centers are distinct.

For a fixed $i$, there are $n-i-2$ choices of $j$. The number of resulting vertices is therefore
\[
    \sum_{i=1}^{n-3}(n-i-2) = \frac{(n-2)(n-3)}2.
\]
\end{proof}

\begin{corollary}\label{cor:tight-nearest-order}
For $n\ge4$, the worst-case combinatorial complexity of the nearest Voronoi diagram of $n$ lines in one ruling of a smooth doubly ruled real quadric, subject to Assumption~\ref{ass:gp}, is
\[
    \max_L \operatorname{comp}(\NVD(L))=\Theta(n^2).
\]
\end{corollary}

\begin{proof}
Theorem~\ref{thm:quadratic-upper-bounds} gives the quadratic upper bound for every input in this class. Theorem~\ref{thm:quadratic-nearest-lower} gives a quadratic lower bound using lines from one ruling of a single fixed surface, since each distinct vertex contributes to the total combinatorial complexity.
\end{proof}

\subparagraph{Algorithmic consequence.}
The support characterization also yields an exact vertex-enumeration algorithm. We count arithmetic operations and comparisons on real numbers as constant-time operations. We additionally allow constant-time computation and exact representation of isolated real solutions of polynomial systems, and polynomial sign determination at those solutions, when the numbers of variables and equations and all polynomial degrees are bounded by absolute constants. Appendix~\ref{app:algorithmic-details} specifies these operations and the output representation in detail.

\begin{corollary}\label{cor:vertex-enumeration}
Let $L$ consist of $n\geq 4$ line sites given by their Pl\"ucker coordinates, promised to belong to one ruling of a smooth doubly ruled real quadric, and satisfying Assumption~\ref{ass:gp}. Under the computational model described above, all finite vertices of $\NVD(L)$ and $\FVD(L)$ can be enumerated exactly in $O(n^2)$ time using $O(n)$ working space, excluding the output. Each vertex is reported once within each diagram. For nearest vertices, the running time is worst-case optimal up to a constant factor.
\end{corollary}

\begin{proof}
Recover a parametrization and the site parameters in $O(n)$ time, as described in Appendix~\ref{app:recovering-conic}, and sort the parameters cyclically in $O(n\log n)$ time. By Lemma~\ref{lem:gap-pair-count}, there are $n(n-3)/2$ unordered pairs of cyclic gaps with no common endpoint. For each pair, solve the three quadratic equations requiring equal distances to its four endpoint lines. Under Assumption~\ref{ass:gp}, every real solution is transverse, so the argument of Lemma~\ref{lem:eight-center-bound} bounds the number of such solutions by eight.

For each candidate, form its contact quartic. Its four distinct support parameters exhaust its roots, so its sign on each intervening arc is determined by one evaluation. Report the candidate as a nearest vertex when the two negative arcs are the selected gaps, and as a farthest vertex when the two positive arcs are the selected gaps. These tests certify the required distance inequalities for every input site without inspecting the sites individually. Lemma~\ref{lem:cyclic-support} ensures completeness. The relevant signed gap pair is uniquely determined by each vertex, which prevents duplicate reports within either diagram.

Each gap pair therefore takes constant time under the stated model, giving $O(n^2)$ total time. Processing one pair at a time and reporting vertices immediately uses $O(n)$ working space. Finally, Theorem~\ref{thm:quadratic-nearest-lower} provides inputs with $\Omega(n^2)$ nearest vertices, so explicitly reporting them requires $\Omega(n^2)$ time.
\end{proof}


\section{Conclusion}\label{sec:conclusion}
Under our general-position assumptions, we have established quadratic bounds on the total combinatorial complexity of nearest and farthest Voronoi diagrams of lines in one ruling of a smooth doubly ruled real quadric. For nearest diagrams, the worst-case complexity over this class is quadratic, and the lower bound already holds for lines in one ruling of a single fixed non-rotational hyperboloid. For farthest diagrams, the total complexity is quadratic for every admissible input, as their three-dimensional cells alone provide a quadratic lower bound.

The proof uses a degree-four polynomial that describes tangency between a sphere and the lines of the ruling. At a Voronoi vertex, the four supporting lines account for all roots of this polynomial. Its alternating signs force the supporting sites to form two pairs of consecutive sites in the cyclic order. This leaves only quadratically many candidate sets of
four supporting lines. The same support characterization also yields exact quadratic-time enumeration of all finite vertices under the computational model used in Corollary~\ref{cor:vertex-enumeration}.

A natural next question is which other families of lines have nearest Voronoi diagrams of quadratic worst-case complexity. Our proof suggests examining how the degree and sign of the sphere-contact polynomial restrict the possible supporting sites. Understanding these restrictions, together with the conditions under which many supporting configurations can be realized, could provide criteria for quadratic complexity in broader families. For arbitrary line sites, improving the general $O(n^{3+\varepsilon})$ upper bound remains an open problem~\cite{Sharir1994}.


\section*{Acknowledgements}
The authors thank Kristian Ranestad for an insightful and stimulating discussion at the AGSTA Conference on Tensors and Related Topics. The discussion helped motivate the algebraic-geometric viewpoint developed in this paper.

\bibliographystyle{plainurl}
\bibliography{references}

@article{AgarwalAronovSharir1997,
  author  = {Agarwal, Pankaj K. and Aronov, Boris and Sharir, Micha},
  title   = {Computing Envelopes in Four Dimensions with Applications},
  journal = {SIAM Journal on Computing},
  volume  = {26},
  number  = {6},
  pages   = {1714--1732},
  year    = {1997},
  doi     = {10.1137/S0097539794265724},
  url     = {https://doi.org/10.1137/S0097539794265724}
}

@article{Aronov2002,
  author  = {Aronov, Boris},
  title   = {A Lower Bound on {Voronoi} Diagram Complexity},
  journal = {Information Processing Letters},
  volume  = {83},
  number  = {4},
  pages   = {183--185},
  year    = {2002},
  doi     = {10.1016/S0020-0190(01)00336-2},
  url     = {https://doi.org/10.1016/S0020-0190(01)00336-2}
}

@article{BarequetPapadopoulouSuderland2024,
  author  = {Barequet, Gill and Papadopoulou, Evanthia and Suderland, Martin},
  title   = {Unbounded Regions of High-Order {Voronoi} Diagrams of Lines and Line Segments in Higher Dimensions},
  journal = {Discrete \& Computational Geometry},
  volume  = {72},
  number  = {3},
  pages   = {1304--1332},
  year    = {2024},
  doi     = {10.1007/s00454-023-00492-2},
  url     = {https://doi.org/10.1007/s00454-023-00492-2}
}

@article{CifuentesEtAl2022,
  author  = {Cifuentes, Diego and Ranestad, Kristian and Sturmfels, Bernd and Weinstein, Madeleine},
  title   = {{Voronoi} Cells of Varieties},
  journal = {Journal of Symbolic Computation},
  volume  = {109},
  pages   = {351--366},
  year    = {2022},
  doi     = {10.1016/j.jsc.2020.07.009},
  url     = {https://doi.org/10.1016/j.jsc.2020.07.009}
}

@article{EverettEtAl2009,
  author  = {Everett, Hazel and Lazard, Daniel and Lazard, Sylvain and Safey El Din, Mohab},
  title   = {The {Voronoi} Diagram of Three Lines},
  journal = {Discrete \& Computational Geometry},
  volume  = {42},
  number  = {1},
  pages   = {94--130},
  year    = {2009},
  doi     = {10.1007/s00454-009-9173-3},
  url     = {https://doi.org/10.1007/s00454-009-9173-3}
}

@misc{Glisse2021,
  author        = {Glisse, Marc},
  title         = {Lower Bound on the {Voronoi} Diagram of Lines in {$\mathbb{R}^d$}},
  year          = {2021},
  eprint        = {2103.17251},
  archivePrefix = {arXiv},
  primaryClass  = {cs.CG},
  doi           = {10.48550/arXiv.2103.17251},
  url           = {https://arxiv.org/abs/2103.17251}
}

@article{GoodmanEtAl2006,
  author  = {Goodman, Jacob E. and Holmsen, Andreas and Pollack, Richard and Ranestad, Kristian and Sottile, Frank},
  title   = {Cremona Convexity, Frame Convexity and a Theorem of {Santal\'o}},
  journal = {Advances in Geometry},
  volume  = {6},
  number  = {2},
  pages   = {301--321},
  year    = {2006},
  doi     = {10.1515/ADVGEOM.2006.018},
  url     = {https://doi.org/10.1515/ADVGEOM.2006.018}
}

@book{Harris1992,
  author    = {Harris, Joe},
  title     = {Algebraic Geometry: A First Course},
  series    = {Graduate Texts in Mathematics},
  volume    = {133},
  publisher = {Springer},
  address   = {New York},
  year      = {1992},
  doi       = {10.1007/978-1-4757-2189-8},
  url       = {https://doi.org/10.1007/978-1-4757-2189-8}
}

@article{KoltunSharir2003,
  author  = {Koltun, Vladlen and Sharir, Micha},
  title   = {{3-Dimensional} Euclidean {Voronoi} Diagrams of Lines with a Fixed Number of Orientations},
  journal = {SIAM Journal on Computing},
  volume  = {32},
  number  = {3},
  pages   = {616--642},
  year    = {2003},
  doi     = {10.1137/S0097539702408387},
  url     = {https://doi.org/10.1137/S0097539702408387}
}

@inproceedings{PapadopoulouWang2026,
  author    = {Papadopoulou, Evanthia and Wang, Zeyu},
  title     = {The {Voronoi} Diagram of Four Lines in {$\mathbb{R}^3$}},
  booktitle = {42nd International Symposium on Computational Geometry (SoCG 2026)},
  series    = {Leibniz International Proceedings in Informatics (LIPIcs)},
  volume    = {367},
  pages     = {84:1--84:17},
  year      = {2026},
  editor    = {Ahn, Hee-Kap and Hoffmann, Michael and Nayyeri, Amir},
  publisher = {Schloss Dagstuhl -- Leibniz-Zentrum f{\"u}r Informatik},
  address   = {Dagstuhl, Germany},
  doi       = {10.4230/LIPIcs.SoCG.2026.84},
  url       = {https://drops.dagstuhl.de/entities/document/10.4230/LIPIcs.SoCG.2026.84}
}

@book{PottmannWallner2001,
  author    = {Pottmann, Helmut and Wallner, Johannes},
  title     = {Computational Line Geometry},
  series    = {Mathematics and Visualization},
  publisher = {Springer},
  address   = {Berlin, Heidelberg},
  year      = {2001},
  doi       = {10.1007/978-3-642-04018-4},
  url       = {https://doi.org/10.1007/978-3-642-04018-4}
}

@article{Sharir1994,
  author  = {Sharir, Micha},
  title   = {Almost Tight Upper Bounds for Lower Envelopes in Higher Dimensions},
  journal = {Discrete \& Computational Geometry},
  volume  = {12},
  number  = {3},
  pages   = {327--345},
  year    = {1994},
  doi     = {10.1007/BF02574384},
  url     = {https://doi.org/10.1007/BF02574384}
}

@article{WeismanChewKedem2004,
  author  = {Weisman, Amit and Chew, L. Paul and Kedem, Klara},
  title   = {{Voronoi} Diagrams of Moving Points in the Plane and of Lines in Space: Tight Bounds for Simple Configurations},
  journal = {Information Processing Letters},
  volume  = {92},
  number  = {5},
  pages   = {245--251},
  year    = {2004},
  doi     = {10.1016/j.ipl.2004.08.004},
  url     = {https://doi.org/10.1016/j.ipl.2004.08.004}
}

@article{Aurenhammer1991,
  author  = {Franz Aurenhammer},
  title   = {Voronoi Diagrams---A Survey of a Fundamental Geometric
             Data Structure},
  journal = {ACM Computing Surveys},
  volume  = {23},
  number  = {3},
  pages   = {345--405},
  year    = {1991},
  doi     = {10.1145/116873.116880}
}

@book{Tu2011,
  author    = {Tu, Loring W.},
  title     = {An Introduction to Manifolds},
  edition   = {2},
  series    = {Universitext},
  publisher = {Springer},
  address   = {New York},
  year      = {2011},
  doi       = {10.1007/978-1-4419-7400-6}
}

@book{basu2006algorithms,
  title={Algorithms in real algebraic geometry},
  author={Basu, Saugata and Pollack, Richard and Roy, Marie-Fran{\c{c}}oise},
  year={2006},
  publisher={Springer}
}

@article{Rudnev2018,
  author  = {Misha Rudnev},
  title   = {On the Number of Incidences Between Points and
             Planes in Three Dimensions},
  journal = {Combinatorica},
  volume  = {38},
  number  = {1},
  pages   = {219--254},
  year    = {2018},
  doi     = {10.1007/s00493-016-3329-6}
}

@book{Ziegler1995,
  author    = {G{\"u}nter M. Ziegler},
  title     = {Lectures on Polytopes},
  series    = {Graduate Texts in Mathematics},
  volume    = {152},
  publisher = {Springer},
  address   = {New York},
  year      = {1995},
  doi       = {10.1007/978-1-4613-8431-1}
}

@misc{PapadopoulouSuderlandWang2026,
  author        = {Evanthia Papadopoulou and Martin Suderland
                   and Zeyu Wang},
  title         = {A Collapse Process for Farthest Voronoi Diagrams
                   of Lines in Three Dimensions},
  year          = {2026},
  eprint        = {2609.29865},
  archivePrefix = {arXiv},
  primaryClass  = {cs.CG},
  url           = {https://arxiv.org/abs/2609.29865}
}

\appendix


\section{Sphere Contacts with Quadric Rulings and Exact Real-Root Counting}\label{app:explicit-rulings}
This appendix gives a coordinate derivation of the contact quartic for an arbitrary ruling of a smooth doubly ruled real quadric. We first write the ruling in homogeneous coordinates and substitute its lines into the equation of a Euclidean sphere. The resulting intersection equation is quadratic along each line, and its discriminant is a homogeneous polynomial of degree four in the ruling parameter. We identify this polynomial with the contact polynomial used in the main text, with the appropriate sign convention. We then describe an exact procedure based on Sturm's theorem for counting its distinct real roots, including a possible root at infinity in the parameter coordinate.

The final subsections give explicit formulas for the one-sheeted hyperboloid and the hyperbolic paraboloid, together with a complete Sturm calculation on the fixed surface used in our lower-bound construction.

\subsection{Coordinates for an arbitrary ruling}
Use homogeneous coordinates $[x_0:x_1:x_2:x_3]$ in $\mathbb P^3(\mathbb R)$. The affine space $\mathbb R^3$ is identified with $x_0=1$, and $x_0=0$ is the plane at infinity.

Let $Q$ be a smooth doubly ruled real quadric. As recalled in Section~\ref{sec:preliminaries}, there is an invertible real $4\times4$ matrix $T$ that maps the standard quadric
\[
    Q_{\mathrm{Seg}} = \bigl\{[y_0:y_1:y_2:y_3] \mid y_0y_3-y_1y_2=0 \bigr\}
\]
onto $Q$.

The standard parametrization is
\[
    ([s:t],[\alpha:\beta]) \longmapsto [s\alpha:s\beta:t\alpha:t\beta].
\]
Fixing $[s:t]$ and varying $[\alpha:\beta]$ gives one ruling. After applying $T$, define
\[
    A(s,t)=T(s,0,t,0)^{\mathsf T}, \qquad B(s,t)=T(0,s,0,t)^{\mathsf T}.
\]
The corresponding projective line is 
\begin{equation}
\label{eq:app-general-ruling-line}
    \overline{\ell}_{[s:t]} = \left\{ [\alpha A(s,t)+\beta B(s,t)] \;\middle|\; [\alpha:\beta]\in\mathbb P^1(\mathbb R) \right\}.
\end{equation}
The vectors $A(s,t)$ and $B(s,t)$ are linearly independent whenever $(s,t)\ne(0,0)$. Their coordinates are linear in $s,t$.
The other ruling is obtained by fixing $[\alpha:\beta]$ and varying $[s:t]$.

The matrix $T$ specifies the lines in the fixed ambient coordinate system. All norms, distances, and sphere equations below use the Euclidean metric in the affine coordinates $(x_1,x_2,x_3)$.

Write
\[
    A=(a_0,a), \qquad B=(b_0,b),
\]
where $a,b\in\mathbb R^3$ and all components depend linearly on $s,t$. The vector--moment coordinates of $\overline{\ell}_{[s:t]}$ are 
\begin{equation}
\label{eq:app-general-vector-moment}
    v(s,t)=a_0b-b_0a, \qquad m(s,t)=a\times b.
\end{equation}
Both vectors have homogeneous quadratic coordinates, and
\[
    v\cdot m = (a_0b-b_0a)\cdot(a\times b) = 0.
\]

For example, if $a_0\ne0$, the line has the affine point $a/a_0$, and
\[
    \frac{a}{a_0}\times v = \frac{a}{a_0}\times(a_0b-b_0a) = a\times b = m.
\]
Thus Equation~\eqref{eq:app-general-vector-moment} agrees with the vector--moment convention of the main text.

For a completely explicit expansion, write the columns of $T$ as
\[
    T_i=(\omega_i,z_i), \qquad \omega_i\in\mathbb R,\quad z_i\in\mathbb R^3, \qquad i=0,1,2,3.
\]
Then
\[
    A=sT_0+tT_2, \qquad B=sT_1+tT_3,
\]
and
\[
    v=v_0s^2+v_1st+v_2t^2, \qquad m=m_0s^2+m_1st+m_2t^2,
\]
where
\begin{align}
    v_0&=\omega_0z_1-\omega_1z_0,
    &
    m_0&=z_0\times z_1,
    \label{eq:app-general-coordinate-coefficients}\\
    v_1&=\omega_0z_3+\omega_2z_1
          -\omega_1z_2-\omega_3z_0,
    &
    m_1&=z_0\times z_3+z_2\times z_1,
    \nonumber\\
    v_2&=\omega_2z_3-\omega_3z_2,
    &
    m_2&=z_2\times z_3.
    \nonumber
\end{align}
These formulas provide quadratic coordinates for either
ruling of every quadric in the class considered here.

\subsection{The line--sphere intersection equation}

Let $C\in\mathbb R^3$ be the center of a sphere and let
$R>0$ be its squared radius.
Writing $x=(x_1,x_2,x_3)$, its homogeneous equation is
\begin{equation}
\label{eq:app-homogeneous-sphere}
    \|x-x_0C\|^2-Rx_0^2=0.
\end{equation}

Substitute a point of
Equation~\eqref{eq:app-general-ruling-line}.
Its homogeneous coordinates are
\[
    x_0=\alpha a_0+\beta b_0,
    \qquad
    x=\alpha a+\beta b.
\]
Define
\[
    \widetilde a=a-a_0C,
    \qquad
    \widetilde b=b-b_0C.
\]
The sphere equation becomes
\[
    \|\alpha\widetilde a+\beta\widetilde b\|^2
    -
    R(\alpha a_0+\beta b_0)^2
    =
    0,
\]
or equivalently
\begin{equation}
\label{eq:app-binary-intersection-equation}
    E(s,t)\alpha^2
    +
    2F(s,t)\alpha\beta
    +
    G(s,t)\beta^2
    =
    0,
\end{equation}
where
\begin{align}
    E&=\|\widetilde a\|^2-Ra_0^2,
    \label{eq:app-intersection-coefficients}\\
    F&=\widetilde a\cdot\widetilde b-Ra_0b_0,
    \nonumber\\
    G&=\|\widetilde b\|^2-Rb_0^2.
    \nonumber
\end{align}
Each of $E,F,G$ is homogeneous of degree two in $s,t$.

The discriminant of the quadratic intersection equation is
\begin{equation}
\label{eq:app-intersection-discriminant}
    \Delta_{C,R}(s,t)
    =
    (2F)^2-4EG
    =
    4(F^2-EG).
\end{equation}
It is therefore homogeneous of degree four, unless it
vanishes identically.

We next identify this discriminant with the contact
polynomial from the main text.
Expanding $EG-F^2$ gives
\begin{align*}
    EG-F^2
    &=
    \|\widetilde a\|^2\|\widetilde b\|^2
    -
    (\widetilde a\cdot\widetilde b)^2\\
    &\quad
    -
    R\bigl(
        a_0^2\|\widetilde b\|^2
        +
        b_0^2\|\widetilde a\|^2
        -
        2a_0b_0\,\widetilde a\cdot\widetilde b
    \bigr)\\
    &=
    \|\widetilde a\times\widetilde b\|^2
    -
    R\|a_0\widetilde b-b_0\widetilde a\|^2.
\end{align*}
The two vector expressions simplify to
\[
    a_0\widetilde b-b_0\widetilde a
    =
    a_0b-b_0a
    =
    v
\]
and
\begin{align*}
    \widetilde a\times\widetilde b
    &=
    (a-a_0C)\times(b-b_0C)\\
    &=
    a\times b-b_0a\times C-a_0C\times b\\
    &=
    m-C\times(a_0b-b_0a)\\
    &=
    m-C\times v.
\end{align*}
Consequently,
\begin{equation}
\label{eq:app-discriminant-contact-identity}
    H_{C,R}(s,t)
    :=
    EG-F^2
    =
    \|C\times v(s,t)-m(s,t)\|^2
    -
    R\|v(s,t)\|^2,
\end{equation}
and
\begin{equation}
\label{eq:app-delta-minus-four-H}
    \Delta_{C,R}=-4H_{C,R}.
\end{equation}

Thus, the contact polynomial used in this paper is negative
one quarter of the intersection discriminant.
This choice gives the sign convention used throughout
the main text.

\subsection{Affine interpretation and parameters at infinity}

Suppose that $v\ne0$, so the projective line has an affine
part.
Set
\[
    D=\|v\|^2,
    \qquad
    a_*=\frac{v\times m}{D}.
\]
The point $a_*$ is the closest point of the line to the
origin.
Indeed, $a_*\cdot v=0$, and
\[
    a_*\times v
    =
    \frac{(v\times m)\times v}{D}
    =
    m,
\]
where we used $v\cdot m=0$.
The affine line can therefore be written as
\[
    x(\lambda)=a_*+\lambda v.
\]

Substituting this expression into
$\|x-C\|^2=R$ gives
\begin{equation}
\label{eq:app-affine-intersection-equation}
    D\lambda^2
    -
    2(C\cdot v)\lambda
    +
    \|a_*-C\|^2-R
    =
    0.
\end{equation}
Its discriminant is
\[
    4(C\cdot v)^2
    -
    4D\bigl(\|a_*-C\|^2-R\bigr)
    =
    -4H_{C,R}.
\]
For example, this identity also follows from
\[
    d(C,\ell)^2
    =
    \|a_*-C\|^2-\frac{(C\cdot v)^2}{D}
\]
and
\[
    H_{C,R}=D\bigl(d(C,\ell)^2-R\bigr).
\]

When $H_{C,R}\le0$, the intersection points are
\begin{equation}
\label{eq:app-sphere-intersection-points}
    x_\pm
    =
    \frac{
        v\times m
        +
        \bigl(C\cdot v\pm\sqrt{-H_{C,R}}\bigr)v
    }{D}.
\end{equation}
It follows that
\begin{align*}
    H_{C,R}<0
    &\iff
    \ell\text{ meets the sphere in two distinct points},\\
    H_{C,R}=0
    &\iff
    \ell\text{ is tangent to the sphere},\\
    H_{C,R}>0
    &\iff
    \ell\text{ misses the closed ball}.
\end{align*}
In the first case, the line also meets the open ball.

A projective generator may instead lie entirely in the
spatial plane at infinity.
In the coordinates above, this happens precisely when
$v=0$.
To see this, suppose that $v=0$ and $a_0\ne0$.
Then $a_0b=b_0a$, which would make $A$ and $B$ proportional.
This contradicts their linear independence.
The same argument applies if $b_0\ne0$.
Hence $a_0=b_0=0$.

In this case $a$ and $b$ are linearly independent,
so $m=a\times b\ne0$, and
\[
    H_{C,R}(s,t)=\|m(s,t)\|^2>0.
\]
A generator lying entirely at spatial infinity therefore
never contributes a real root of the contact polynomial.

The point at infinity of a chosen parameter coordinate
has a different role.
It can represent an ordinary affine generator and must
be included when counting roots on the parameter circle.

\subsection{The five coefficients of the contact quartic}

Using
\[
    v=v_0s^2+v_1st+v_2t^2,
    \qquad
    m=m_0s^2+m_1st+m_2t^2,
\]
put
\[
    n_i=C\times v_i-m_i,
    \qquad i=0,1,2.
\]
Then
\[
    C\times v-m=n_0s^2+n_1st+n_2t^2.
\]
Expanding the squared norms in
Equation~\eqref{eq:app-discriminant-contact-identity}
gives
\begin{equation}
\label{eq:app-general-contact-expansion}
    H_{C,R}(s,t)
    =
    h_0s^4+h_1s^3t+h_2s^2t^2+h_3st^3+h_4t^4,
\end{equation}
where
\begin{align}
    h_0
    &=
    \|n_0\|^2-R\|v_0\|^2,
    \label{eq:app-general-contact-coefficients}\\
    h_1
    &=
    2\bigl(n_0\cdot n_1-Rv_0\cdot v_1\bigr),
    \nonumber\\
    h_2
    &=
    \|n_1\|^2+2n_0\cdot n_2
    -
    R\bigl(\|v_1\|^2+2v_0\cdot v_2\bigr),
    \nonumber\\
    h_3
    &=
    2\bigl(n_1\cdot n_2-Rv_1\cdot v_2\bigr),
    \nonumber\\
    h_4
    &=
    \|n_2\|^2-R\|v_2\|^2.
    \nonumber
\end{align}

Equations~\eqref{eq:app-general-coordinate-coefficients}
and~\eqref{eq:app-general-contact-coefficients} give all
coefficients directly from $T$, the center $C$, and the
squared radius $R$.
Every coefficient is polynomial in these data.

If all five coefficients vanish, then every affine
generator of the ruling is tangent to the sphere.
This is the case excluded by
Assumption~\ref{ass:gp}(5).
Otherwise, $H_{C,R}$ is a homogeneous polynomial of
degree four.

\subsection{Counting real contact parameters with Sturm's theorem}

Fix a sphere and abbreviate $H=H_{C,R}$.
We first test whether all coefficients in
Equation~\eqref{eq:app-general-contact-expansion} vanish.
The identically zero case is handled separately as above.
For the rest of this subsection, assume $H\not\equiv0$.

\subsubsection{Finite parameters and the root at infinity}

In the real parameter coordinate $u=t/s$, define
\[
    h(u)=H(1,u)
    =
    h_0+h_1u+h_2u^2+h_3u^3+h_4u^4.
\]
Every real root $u$ corresponds to the contact parameter
$[1:u]$.
Let $d=\deg h$, using its actual degree after zero leading
coefficients have been removed.

The remaining parameter is $[0:1]$.
Using $w=s/t$ near this point gives
\[
    H(w,1)=w^4h(1/w).
\]
If the leading coefficient of $h$ is $h_d$, then
\[
    H(w,1)=h_dw^{4-d}+\text{higher powers of }w.
\]
Consequently, $[0:1]$ is a root if and only if $d<4$,
and its multiplicity is exactly $4-d$.

In particular, a drop in the degree of $h$ does not reduce
the homogeneous degree of $H$.
It records a root at the omitted parameter value.

\subsubsection{Removing repeated factors}

If $d=0$, there are no finite roots.
Suppose now that $d\ge1$.
Compute the greatest common divisor
\[
    g(u)=\gcd(h(u),h'(u)),
\]
normalized to have leading coefficient one, and define
\[
    f(u)=\frac{h(u)}{g(u)}.
\]
The polynomial $f$ has exactly the same distinct roots
as $h$, and none of its roots is repeated.
Such a polynomial is called \emph{squarefree}.

This step allows us to count distinct contact parameters
even when $h$ has multiple roots.
If $g$ is nonconstant, at least one finite complex root of $h$ is repeated.
Repeated roots at the omitted parameter $[0:1]$ are
detected separately by the condition $4-d\ge2$.

\subsubsection{The Sturm sequence}

Define
\[
    S_0=f,
    \qquad
    S_1=f',
\]
and repeatedly compute
\begin{equation}
\label{eq:app-sturm-recurrence}
    S_{j+1}
    =
    -\operatorname{rem}(S_{j-1},S_j),
\end{equation}
where $\operatorname{rem}(A,B)$ denotes the remainder
obtained by polynomial division of $A$ by $B$.
Stop when the next remainder is zero, and omit that zero
polynomial from the sequence.

The degree decreases at every nonzero remainder.
Because $\deg f\le4$, only a bounded number of divisions
is required.
Since $f$ is squarefree, the final nonzero polynomial
is a constant.

For a real number $a$, let $\mathcal V(a)$ be the number
of sign changes in
\[
    S_0(a),S_1(a),\ldots,S_m(a),
\]
after zero entries are removed.
For example, the list $(+,-,0,-,+)$ has two sign changes.

Sturm's theorem states that, if $a<b$ and neither endpoint
is a root of $f$, then
\begin{equation}
\label{eq:app-sturm-interval-count}
    \#\{u\in(a,b):h(u)=0\}
    =
    \mathcal V(a)-\mathcal V(b),
\end{equation}
where roots are counted without multiplicity.
See Basu, Pollack, and Roy
\cite[Theorem~2.50]{basu2006algorithms}.

The formula also applies to infinite endpoints.
If $S$ is a nonzero polynomial of degree $e$ and leading
coefficient $\alpha$, then
\[
    \operatorname{sgn}S(+\infty)
    =
    \operatorname{sgn}\alpha,
    \qquad
    \operatorname{sgn}S(-\infty)
    =
    (-1)^e\operatorname{sgn}\alpha.
\]
Thus, the number of distinct finite real roots is
\begin{equation}
\label{eq:app-sturm-total-finite}
    N_{\mathrm{fin}}
    =
    \mathcal V(-\infty)-\mathcal V(+\infty).
\end{equation}
For $d=0$, set $N_{\mathrm{fin}}=0$.

Applying Equation~\eqref{eq:app-sturm-interval-count}
to smaller intervals also allows individual finite roots
to be isolated.
A polynomial together with an interval containing exactly
one of its real roots provides an exact description of
that root.

\subsubsection{An explicit first division for a quartic}

When $h$ is squarefree and $h_4\ne0$, the first two
Sturm polynomials are
\[
    S_0
    =
    h_4u^4+h_3u^3+h_2u^2+h_1u+h_0
\]
and
\[
    S_1
    =
    4h_4u^3+3h_3u^2+2h_2u+h_1.
\]
Direct polynomial division gives
\begin{align}
    S_2
    &=
    \frac{3h_3^2-8h_4h_2}{16h_4}\,u^2
    +
    \frac{h_3h_2-6h_4h_1}{8h_4}\,u
    +
    \frac{h_3h_1-16h_4h_0}{16h_4}.
    \label{eq:app-sturm-first-remainder}
\end{align}

The degree of $S_2$ may be smaller than two.
Subsequent divisions use its actual degree.
Likewise, the full sequence need not have degrees
$4,3,2,1,0$.
The recurrence~\eqref{eq:app-sturm-recurrence} handles
these degree drops directly.

Multiplying an individual Sturm polynomial by a positive
constant leaves its signs unchanged.
Any simplification of the sequence must preserve these
signs.

\subsubsection{The projective root count}

Combining the finite count with the omitted parameter gives
\begin{equation}
\label{eq:app-sturm-projective-count}
    N_{\mathrm{proj}}
    =
    N_{\mathrm{fin}}
    +
    \begin{cases}
        1,&d<4,\\
        0,&d=4.
    \end{cases}
\end{equation}
As shown above, a real root cannot represent a generator
lying entirely in the spatial plane at infinity.
Thus $N_{\mathrm{proj}}$ is exactly the number of distinct
affine generators tangent to the sphere.

\begin{proposition}\label{prop:app-four-real-contact-test}
Let $H\not\equiv0$ be the contact polynomial of a sphere
along a ruling.
The sphere is tangent to four distinct generators if and
only if the count in
Equation~\eqref{eq:app-sturm-projective-count} equals four.
All four contact parameters are then simple roots.

Equivalently, one of the following holds.
\begin{enumerate}
    \item
    The polynomial $h$ has degree four, has no repeated
    roots, and its Sturm count is four.

    \item
    The polynomial $h$ has degree three, has no repeated
    roots, and its Sturm count is three.
    The fourth contact parameter is $[0:1]$.
\end{enumerate}
\end{proposition}

\begin{proof}
Equation~\eqref{eq:app-sturm-projective-count} counts
all distinct real projective roots.
A nonzero homogeneous quartic has total root multiplicity
four over the complex projective line.
Hence four distinct real roots exhaust this multiplicity,
and each is simple.

If the omitted parameter is a root, simplicity requires
$4-d=1$, so $d=3$.
The remaining three roots must be distinct and real.
If the omitted parameter is not a root, $d=4$ and all
four roots are finite.
\end{proof}

The squarefree polynomial $f$ is used to count roots.
For geometric sign tests, we continue to use the original
polynomial $H$.
Removing repeated factors can change the signs between
roots.

The procedure uses arithmetic operations and exact sign
tests on expressions in the coefficients $h_0,\ldots,h_4$.
On each branch of the computation, its divisions have
nonzero denominators.
Clearing these denominators by positive squared factors
converts the tests into polynomial equalities and
inequalities.
Thus the sphere parameters giving four distinct real
contacts form a set described by a finite Boolean
combination of polynomial equalities and inequalities.
Such a set is called \emph{semialgebraic}.

\subsection{Explicit formulas for one-sheeted hyperboloids}

For $a,b,c>0$, consider
\[
    \mathcal H_{a,b,c}
    =
    \left\{
        (x,y,z)\in\mathbb R^3
        \;\middle|\;
        \frac{x^2}{a^2}+\frac{y^2}{b^2}
        -\frac{z^2}{c^2}=1
    \right\}.
\]
Its two rulings are parametrized by
\begin{align*}
    \ell_\theta^+(\lambda)
    &=
    \bigl(
        a(\cos\theta-\lambda\sin\theta),
        b(\sin\theta+\lambda\cos\theta),
        c\lambda
    \bigr),\\
    \ell_\theta^-(\lambda)
    &=
    \bigl(
        a(\cos\theta+\lambda\sin\theta),
        b(\sin\theta-\lambda\cos\theta),
        c\lambda
    \bigr),
\end{align*}
where angles are taken modulo $2\pi$.
The parameter $\theta$ selects a line, and $\lambda$
specifies a point on it.

Substitution gives
\[
    (\cos\theta-\lambda\sin\theta)^2
    +
    (\sin\theta+\lambda\cos\theta)^2
    -
    \lambda^2
    =
    1,
\]
and the same identity holds with the opposite signs.
Thus all these lines lie on the hyperboloid.

To verify that they cover the surface, put
\[
    \xi=\frac{x}{a},
    \qquad
    \eta=\frac{y}{b},
    \qquad
    \zeta=\frac{z}{c}.
\]
For a point of the hyperboloid,
$\xi^2+\eta^2=1+\zeta^2$.
The angle of the unique line in the plus family through
this point satisfies
\[
    \cos\theta_+
    =
    \frac{\xi+\zeta\eta}{1+\zeta^2},
    \qquad
    \sin\theta_+
    =
    \frac{\eta-\zeta\xi}{1+\zeta^2}.
\]
For the minus family,
\[
    \cos\theta_-
    =
    \frac{\xi-\zeta\eta}{1+\zeta^2},
    \qquad
    \sin\theta_-
    =
    \frac{\eta+\zeta\xi}{1+\zeta^2}.
\]
In each case, the sum of the squares of the two expressions
is one.
Every surface point therefore belongs to exactly one line
of each family.

Distinct lines in either family cannot intersect at an
affine point by this uniqueness.
Their direction vectors have the same nonzero third
coordinate $c$.
Two such vectors are proportional only when the
proportionality factor is one, which forces the angles
to coincide modulo $2\pi$.
Thus distinct lines in the same family are also nonparallel
and hence skew.

For lines from opposite families, set $\delta=\varphi-\theta$.
If $\delta\not\equiv\pi\pmod{2\pi}$, substitution shows that
$\ell_\theta^+$ and $\ell_\varphi^-$ intersect at the common
line parameter
\[
    \lambda=\tan(\delta/2).
\]
If $\delta\equiv\pi\pmod{2\pi}$, the affine lines are parallel,
and their projective closures intersect at their common
point at infinity.

\subsubsection{Quadratic coordinates for the ruling}

For the plus family, use
\[
    u=\tan(\theta/2)=t/s.
\]
The half-angle identities give
\[
    \cos\theta=\frac{s^2-t^2}{s^2+t^2},
    \qquad
    \sin\theta=\frac{2st}{s^2+t^2}.
\]
Multiplying the direction and moment by $s^2+t^2$ yields
\begin{align}
    V(s,t)
    &=
    \bigl(
        -2ast,\,
        b(s^2-t^2),\,
        c(s^2+t^2)
    \bigr),
    \label{eq:app-hyperboloid-vector-moment}\\
    M(s,t)
    &=
    \bigl(
        2bcst,\,
        -ac(s^2-t^2),\,
        ab(s^2+t^2)
    \bigr).
    \nonumber
\end{align}
Their scalar product is zero.

The minus ruling has
\[
    V^-(s,t)
    =
    \bigl(
        2ast,\,
        -b(s^2-t^2),\,
        c(s^2+t^2)
    \bigr)
\]
and
\[
    M^-(s,t)
    =
    \bigl(
        2bcst,\,
        -ac(s^2-t^2),\,
        -ab(s^2+t^2)
    \bigr).
\]
All real parameter values represent affine lines because
the third component of either direction vector is
$c(s^2+t^2)>0$.

\subsubsection{The explicit contact quartic}

Let the sphere center be $C=(X,Y,Z)$.
For the plus ruling, write
\[
    C\times V-M=W_0s^2+W_1st+W_2t^2,
\]
where
\begin{align*}
    W_0&=(cY-bZ,\,c(a-X),\,b(X-a)),\\
    W_1&=(-2bc,\,-2aZ,\,2aY),\\
    W_2&=(cY+bZ,\,-c(a+X),\,-b(a+X)).
\end{align*}
Also,
\[
    \|V(s,t)\|^2
    =
    (b^2+c^2)(s^4+t^4)
    +
    (4a^2-2b^2+2c^2)s^2t^2.
\]

Consequently,
\[
    H_{C,R}(s,t)
    =
    \sum_{k=0}^4h_ks^{4-k}t^k,
\]
with
\begin{align}
    h_0
    &=
    (cY-bZ)^2
    +
    (b^2+c^2)\bigl((X-a)^2-R\bigr),
    \label{eq:app-hyperboloid-contact-coefficients}\\
    h_1
    &=
    4\bigl[
        b\bigl(a(X-a)-c^2\bigr)Y
        +
        c\bigl(b^2+a(X-a)\bigr)Z
    \bigr],
    \nonumber\\
    h_2
    &=
    2(c^2-b^2)(X^2-a^2)
    +(4a^2+2c^2)Y^2
    +(4a^2-2b^2)Z^2
    +4b^2c^2
    \nonumber\\
    &\qquad
    -R(4a^2-2b^2+2c^2),
    \nonumber\\
    h_3
    &=
    4\bigl[
        -b\bigl(c^2+a(X+a)\bigr)Y
        +
        c\bigl(a(X+a)-b^2\bigr)Z
    \bigr],
    \nonumber\\
    h_4
    &=
    (cY+bZ)^2
    +
    (b^2+c^2)\bigl((X+a)^2-R\bigr).
    \nonumber
\end{align}
The contact polynomial for the minus ruling is obtained
by replacing $Z$ with $-Z$.

These formulas apply to every positive choice of $a,b,c$.
For the rotational case $a=b$ and center $C=0$, they reduce to
\[
    H_{0,R}(s,t)
    =
    (a^2-R)(a^2+c^2)(s^2+t^2)^2.
\]
In particular, $R=a^2$ makes the polynomial identically zero.
The sphere centered at the origin with radius $a$ is then
tangent to every line of either ruling.

\subsection{Explicit formulas for the hyperbolic paraboloid}

Consider
\[
    \mathcal P
    =
    \bigl\{
        (x,y,z)\in\mathbb R^3
        \mid z=x^2-y^2
    \bigr\}.
\]
The factorization
\[
    z=(x-y)(x+y)
\]
gives the two affine line families
\begin{align*}
    \lambda_\alpha(\rho)
    &=
    \left(
        \frac{\rho+\alpha}{2},
        \frac{\rho-\alpha}{2},
        \alpha\rho
    \right),\\
    \mu_\beta(\rho)
    &=
    \left(
        \frac{\rho+\beta}{2},
        \frac{\beta-\rho}{2},
        \beta\rho
    \right).
\end{align*}
Along $\lambda_\alpha$, the coordinate $x-y$ is fixed
at $\alpha$ and $x+y=\rho$.
Along $\mu_\beta$, the coordinate $x+y$ is fixed at
$\beta$ and $x-y=\rho$.
Every point of $\mathcal P$ belongs to exactly one line
of each family.

Two distinct lines $\lambda_\alpha$ and
$\lambda_{\alpha'}$ lie in the disjoint parallel planes
$x-y=\alpha$ and $x-y=\alpha'$.
Their direction vectors can be chosen as
$(1,1,2\alpha)$ and $(1,1,2\alpha')$, which are parallel
only when $\alpha=\alpha'$.
Thus distinct lines in this family are skew.
The same argument applies to the other family, whose
direction vectors are $(1,-1,2\beta)$.

A line from each family intersects at
\[
    \lambda_\alpha\cap\mu_\beta
    =
    \left\{
        \left(
            \frac{\alpha+\beta}{2},
            \frac{\beta-\alpha}{2},
            \alpha\beta
        \right)
    \right\}.
\]

For the first family, a point and direction vector are
\[
    a_\alpha
    =
    \left(\frac{\alpha}{2},-\frac{\alpha}{2},0\right),
    \qquad
    v_\alpha=(1,1,2\alpha).
\]
Their moment is
\[
    m_\alpha=a_\alpha\times v_\alpha
    =
    (-\alpha^2,-\alpha^2,\alpha).
\]
Writing $\alpha=t/s$ and homogenizing to degree two gives
\begin{equation}
\label{eq:app-paraboloid-vector-moment}
    V(s,t)=(s^2,s^2,2st),
    \qquad
    M(s,t)=(-t^2,-t^2,st).
\end{equation}
For $C=(X,Y,Z)$,
\[
    C\times V-M
    =
    \bigl(
        t^2+2Yst-Zs^2,\,
        t^2-2Xst+Zs^2,\,
        (X-Y)s^2-st
    \bigr),
\]
and
\[
    \|V(s,t)\|^2=2s^4+4s^2t^2.
\]
Expanding gives
\begin{align}
    H_{C,R}(s,t)
    &=
    \bigl(2Z^2+(X-Y)^2-2R\bigr)s^4
    \label{eq:app-paraboloid-contact}\\
    &\quad
    +\bigl(-4Z(X+Y)-2(X-Y)\bigr)s^3t
    \nonumber\\
    &\quad
    +\bigl(4(X^2+Y^2)+1-4R\bigr)s^2t^2
    \nonumber\\
    &\quad
    +4(Y-X)st^3+2t^4.
    \nonumber
\end{align}
For the other family, one may use
\[
    V^-(s,t)=(s^2,-s^2,2st),
    \qquad
    M^-(s,t)=(t^2,-t^2,-st).
\]
Its contact polynomial is obtained from
Equation~\eqref{eq:app-paraboloid-contact} by replacing
$Y$ with $-Y$.

The projective closure of $\mathcal P$ has equation
\[
    x_0x_3=x_1^2-x_2^2.
\]
Its intersection with $x_0=0$ consists of the two lines
\[
    x_0=0,\quad x_1-x_2=0
    \qquad\text{and}\qquad
    x_0=0,\quad x_1+x_2=0.
\]
There is one such line in each ruling.
For the first family, the omitted parameter $[0:1]$
represents its generator at spatial infinity.
Indeed, Equation~\eqref{eq:app-paraboloid-vector-moment}
gives $V(0,1)=0$, and
\[
    H_{C,R}(0,1)=2>0.
\]
Thus all real contact parameters for this family are
finite in the coordinate $\alpha=t/s$.

The quartic calculation remains valid for this surface.
However, all direction vectors in the first ruling lie
in the plane $v_x-v_y=0$, and all direction vectors in
the second lie in $v_x+v_y=0$.
These rulings therefore fail the common-plane restriction
in Assumption~\ref{ass:gp}(1).

\subsection{A complete Sturm calculation on the fixed hyperboloid}

Set $a=2$ and $b=c=1$.
This gives the fixed surface
\[
    Q_*:\frac{x^2}{4}+y^2-z^2=1
\]
used in Section~\ref{sec:lower-bound}.

For a sphere centered at the origin,
Equation~\eqref{eq:app-hyperboloid-contact-coefficients}
reduces to
\begin{equation}
\label{eq:app-centered-contact-example}
    H_{0,R}(s,t)
    =
    (8-2R)(s^4+t^4)
    +
    (4-16R)s^2t^2.
\end{equation}
Take $R=2$.
Then
\[
    H_{0,2}(s,t)
    =
    4(s^4-7s^2t^2+t^4).
\]
Since the factor $4$ is positive, it suffices to use
\[
    h(u)=u^4-7u^2+1,
    \qquad u=t/s.
\]

The negative-remainder computation gives
\begin{align*}
    S_0(u)&=u^4-7u^2+1,\\
    S_1(u)&=4u^3-14u,\\
    S_2(u)&=\frac72u^2-1,\\
    S_3(u)&=\frac{90}{7}u,\\
    S_4(u)&=1.
\end{align*}
For example,
\[
    S_0
    =
    \frac{u}{4}S_1
    -
    \left(\frac72u^2-1\right),
\]
and
\[
    S_1
    =
    \frac{8u}{7}S_2-\frac{90}{7}u.
\]
The next remainder is $-1$, which gives $S_4=1$.

The signs at the two infinite endpoints are
\[
\begin{array}{c|ccccc|c}
    &S_0&S_1&S_2&S_3&S_4&\mathcal V\\ \hline
    -\infty&+&-&+&-&+&4\\
    +\infty&+&+&+&+&+&0
\end{array}
\]
and hence
\[
    N_{\mathrm{fin}}=4-0=4.
\]
There is no additional root at the omitted parameter because
\[
    H_{0,2}(0,1)=4\ne0.
\]
The sphere is therefore tangent to four distinct generators.
The four roots can also be written explicitly as
\[
    u=\pm\sqrt{\frac{7+3\sqrt5}{2}},
    \qquad
    u=\pm\sqrt{\frac{7-3\sqrt5}{2}}.
\]

Other sphere radii give different root patterns on the
same ruling.
For example,
\[
    H_{0,1/4}(s,t)
    =
    \frac{15}{2}(s^4+t^4)>0
\]
at every real projective parameter, so this sphere has
no tangent generator.
For $R=1$,
\[
    H_{0,1}(s,t)=6(s^2-t^2)^2,
\]
which has two distinct real roots, each of multiplicity two.

The sphere used in the lower-bound construction satisfies
\[
    H_{(5/2,0,0),\,1/4}(s,t)=40t^4.
\]
Its contact parameter $[1:0]$ has multiplicity four.
Reflecting the center gives
\[
    H_{(-5/2,0,0),\,1/4}(s,t)=40s^4.
\]
In the coordinate $u=t/s$, the latter polynomial becomes
the nonzero constant $40$.
It has no finite root, but the omitted parameter $[0:1]$
is a root of multiplicity four.
This example shows why the projective root count must
include that parameter.

\subsection{Relation to the Voronoi arguments}

The coordinate calculation establishes the homogeneous
degree-four contact equation for every ruling of a smooth
doubly ruled real quadric.
Sturm's theorem then determines how many distinct real
contact parameters a specified sphere has.

In the upper-bound argument, the sphere is already known
to be tangent to four supporting sites.
These four real parameters exhaust the roots of the
nonzero quartic.
Their simplicity and the alternating signs follow
directly from its degree.

In the lower-bound construction, the four desired real
contact parameters are prescribed first.
The coefficient map then supplies a sphere realizing
the polynomial with those roots.
The Sturm procedure gives an independent exact test
for any sphere obtained by this construction or by
another computation.

Finally, the existence of four real contact parameters
alone does not determine whether the sphere supports
a nearest or farthest Voronoi vertex.
The signs of the original contact polynomial at the
remaining site parameters determine the required distance
inequalities.
Regularity is verified by the independence of the three
distance-equality gradients, as explained in the proof of
Lemma~\ref{lem:eight-center-bound}.


\section{Technical Details for the Lower Bound}
\label{app:lower-bound-details}

This appendix supplies the calculations and general-position
arguments used in Section~\ref{sec:lower-bound}.
We first compute the contact polynomial for an explicit family of
hyperboloids and find a sphere for which this polynomial is a
positive multiple of $u^4$.
We then show that small changes in its four lower coefficients
can be realized by nearby spheres.
Finally, we choose the site parameters so that all the required
general-position conditions hold, while preserving every sphere
used in the lower-bound construction.

\subsection{Expansion of the contact quartic}
\label{app:contact-expansion}

For $\lambda>1$, consider the hyperboloid
\[
    Q_\lambda:
    \frac{x^2}{\lambda^2}+y^2-z^2=1.
\]
For a finite real parameter $u$, define
\[
    a_\lambda(u)
    =
    \left(
        \frac{\lambda(1-u^2)}{1+u^2},
        \frac{2u}{1+u^2},
        0
    \right)
\]
and
\[
    V_\lambda(u)
    =
    (-2\lambda u,\,1-u^2,\,1+u^2).
\]
The lines
\[
    \ell_u^\lambda
    =
    a_\lambda(u)+\mathbb R V_\lambda(u)
\]
belong to one ruling of $Q_\lambda$.
This parametrization is obtained from the trigonometric
parametrization in Appendix~\ref{app:explicit-rulings} by setting
$u=\tan(\theta/2)$.

The moment vector associated with this choice of direction is
\[
    M_\lambda(u)
    =
    a_\lambda(u)\times V_\lambda(u)
    =
    \bigl(
        2u,\,-\lambda(1-u^2),\,\lambda(1+u^2)
    \bigr).
\]
In particular,
\[
    V_\lambda(u)\cdot M_\lambda(u)=0.
\]

Let $c=(X,Y,Z)$ be a sphere center and let $R$ denote its squared
radius.
Its contact polynomial along this ruling is
\[
    P^\lambda_{c,R}(u)
    =
    \|c\times V_\lambda(u)-M_\lambda(u)\|^2
    -
    R\|V_\lambda(u)\|^2.
\]
By the distance formula,
\[
    P^\lambda_{c,R}(u)
    =
    \|V_\lambda(u)\|^2
    \bigl(d(c,\ell_u^\lambda)^2-R\bigr).
\]
The factor $\|V_\lambda(u)\|^2$ is positive for every real $u$.
Thus the sign of the contact polynomial determines whether the
line meets the open ball, is tangent to its boundary, or misses
the closed ball.

To expand the polynomial, write
\[
    c\times V_\lambda(u)-M_\lambda(u)
    =
    W_0+W_1u+W_2u^2,
\]
where
\begin{align*}
    W_0
    &=
    (Y-Z,\,\lambda-X,\,X-\lambda),\\
    W_1
    &=
    (-2,\,-2\lambda Z,\,2\lambda Y),\\
    W_2
    &=
    (Y+Z,\,-X-\lambda,\,-X-\lambda).
\end{align*}
Also,
\begin{align*}
    \|V_\lambda(u)\|^2
    &=
    4\lambda^2u^2+(1-u^2)^2+(1+u^2)^2\\
    &=
    2+4\lambda^2u^2+2u^4.
\end{align*}

Writing
\[
    P^\lambda_{c,R}(u)
    =
    \sum_{k=0}^{4}p_k^\lambda(c,R)u^k,
\]
we therefore obtain
\begin{align*}
    p_0^\lambda
    &=
    \|W_0\|^2-2R,\\
    p_1^\lambda
    &=
    2W_0\cdot W_1,\\
    p_2^\lambda
    &=
    \|W_1\|^2+2W_0\cdot W_2-4\lambda^2R,\\
    p_3^\lambda
    &=
    2W_1\cdot W_2,\\
    p_4^\lambda
    &=
    \|W_2\|^2-2R.
\end{align*}
Substituting the three vectors gives
\begin{align}
    p_0^\lambda
    &=
    (Y-Z)^2+2(X-\lambda)^2-2R,
    \label{eq:plambda0}\\
    p_1^\lambda
    &=
    4\bigl(
        -(Y-Z)+\lambda(X-\lambda)(Y+Z)
    \bigr),
    \label{eq:plambda1}\\
    p_2^\lambda
    &=
    4+(4\lambda^2+2)Y^2
      +(4\lambda^2-2)Z^2
      -4\lambda^2R,
    \label{eq:plambda2}\\
    p_3^\lambda
    &=
    4\bigl(
        -(Y+Z)+\lambda(X+\lambda)(Z-Y)
    \bigr),
    \label{eq:plambda3}\\
    p_4^\lambda
    &=
    (Y+Z)^2+2(X+\lambda)^2-2R.
    \label{eq:plambda4}
\end{align}

Now set
\[
    c_\lambda
    =
    (\lambda+\lambda^{-1},0,0),
    \qquad
    R_\lambda=\lambda^{-2}.
\]
At these values,
\[
    p_0^\lambda
    =
    2\lambda^{-2}-2\lambda^{-2}
    =
    0,
    \qquad
    p_2^\lambda
    =
    4-4\lambda^2\lambda^{-2}
    =
    0.
\]
The coefficients $p_1^\lambda$ and $p_3^\lambda$ vanish because
$Y=Z=0$.
Finally,
\begin{align*}
    p_4^\lambda
    &=
    2(2\lambda+\lambda^{-1})^2
    -
    2\lambda^{-2}\\
    &=
    8(\lambda^2+1).
\end{align*}
Consequently,
\[
    P^\lambda_{c_\lambda,R_\lambda}(u)
    =
    8(\lambda^2+1)u^4.
\]

For $\lambda=2$, this gives
\[
    c_*=\left(\frac52,0,0\right),
    \qquad
    R_*=\frac14,
    \qquad
    P_{c_*,R_*}(u)=40u^4.
\]
Here the root of multiplicity four concerns the dependence of
the contact polynomial on the ruling parameter $u$.
It provides the polynomial from which we construct nearby
polynomials with four distinct roots.

\subsection{The coefficient Jacobian}
\label{app:coefficient-jacobian}

Near $(c_\lambda,R_\lambda)$, the coefficient $p_4^\lambda$
remains positive. Normalizing by this coefficient gives the map
\[
    \Psi_\lambda(X,Y,Z,R)
    =
    \left(
        \frac{p_0^\lambda}{p_4^\lambda},
        \frac{p_1^\lambda}{p_4^\lambda},
        \frac{p_2^\lambda}{p_4^\lambda},
        \frac{p_3^\lambda}{p_4^\lambda}
    \right).
\]
At the chosen sphere,
\[
    \Psi_\lambda(c_\lambda,R_\lambda)
    =
    (0,0,0,0).
\]

We compute the matrix of first derivatives of this map, with
columns ordered as $(X,Y,Z,R)$.
Write
\[
    A=\lambda^2+1.
\]
Since $p_k^\lambda=0$ for $k=0,1,2,3$ and $p_4^\lambda=8A$
at the chosen sphere, the quotient rule gives
\[
    D\left(\frac{p_k^\lambda}{p_4^\lambda}\right)
    =
    \frac{Dp_k^\lambda}{8A},
    \qquad k=0,1,2,3,
\]
where all derivatives in this equation are evaluated at
$(c_\lambda,R_\lambda)$.

Differentiating
Equations~\eqref{eq:plambda0}--\eqref{eq:plambda3} gives
\[
    D(p_0^\lambda,p_1^\lambda,p_2^\lambda,p_3^\lambda)
    =
    \begin{pmatrix}
        4/\lambda & 0 & 0 & -2\\
        0 & 0 & 8 & 0\\
        0 & 0 & 0 & -4\lambda^2\\
        0 & -8A & 8\lambda^2 & 0
    \end{pmatrix}
\]
at $(c_\lambda,R_\lambda)$.
Dividing each row by $8A$, we obtain
\[
    D\Psi_\lambda(c_\lambda,R_\lambda)
    =
    \begin{pmatrix}
        \dfrac{1}{2\lambda A}
            & 0 & 0 & -\dfrac{1}{4A}\\[6pt]
        0 & 0 & \dfrac{1}{A} & 0\\[6pt]
        0 & 0 & 0 & -\dfrac{\lambda^2}{2A}\\[6pt]
        0 & -1 & \dfrac{\lambda^2}{A} & 0
    \end{pmatrix}.
\]
Its determinant is
\[
    \det D\Psi_\lambda(c_\lambda,R_\lambda)
    =
    \frac{\lambda}{4(\lambda^2+1)^3}
    \neq0.
\]

In particular, for $\lambda=2$,
\[
    D\Psi(c_*,R_*)
    =
    \begin{pmatrix}
        \dfrac1{20} & 0 & 0 & -\dfrac1{20}\\[4pt]
        0 & 0 & \dfrac15 & 0\\[4pt]
        0 & 0 & 0 & -\dfrac25\\[4pt]
        0 & -1 & \dfrac45 & 0
    \end{pmatrix},
    \qquad
    \det D\Psi(c_*,R_*)=\frac1{250}.
\]

The inverse function theorem now supplies an open neighborhood
$\mathcal W$ of $(c_*,R_*)$ and an open neighborhood
$\mathcal U$ of the origin in coefficient space such that
\[
    \Psi:\mathcal W\longrightarrow\mathcal U
\]
is one-to-one and has a smooth inverse.
After shrinking these neighborhoods, we may also require
\[
    R>0,
    \qquad
    p_4>0,
    \qquad
    \det D\Psi(c,R)\neq0
\]
throughout $\mathcal W$.

It follows that every coefficient vector
$(a_0,a_1,a_2,a_3)\in\mathcal U$ determines a unique sphere
$(c,R)\in\mathcal W$ satisfying
\[
    P_{c,R}(u)
    =
    p_4(c,R)
    \bigl(
        u^4+a_3u^3+a_2u^2+a_1u+a_0
    \bigr).
\]
The uniqueness here is within $\mathcal W$.
This proves Corollary~\ref{cor:local-quartic-realization}.

Unless $\lambda$ is explicitly displayed, the notation in the
rest of this appendix refers to $\lambda=2$.
Thus
\[
    \ell_u=\ell_u^2,
    \qquad
    V(u)=V_2(u),
    \qquad
    P_{c,R}=P^2_{c,R},
    \qquad
    \Psi=\Psi_2.
\]

\subsection{Regularity and transversality}
\label{app:regularity-transversality}

We next relate the derivative of the coefficient map to the
geometry of a sphere tangent to four lines.
The calculation applies to any such sphere with $p_4\neq0$,
including spheres outside the neighborhood $\mathcal W$.

Fix four distinct real parameters $r_1,r_2,r_3,r_4$.
For each $k$, define
\[
    F_k(c,R)
    =
    d(c,\ell_{r_k})^2-R.
\]
The equation $F_k(c,R)=0$ says that the sphere is tangent to
$\ell_{r_k}$.

Write
\[
    \Psi(c,R)=(a_0,a_1,a_2,a_3).
\]
The contact-distance identity gives
\[
    F_k(c,R)
    =
    \frac{p_4(c,R)}{\|V(r_k)\|^2}
    \bigl(
        r_k^4+a_3r_k^3+a_2r_k^2+a_1r_k+a_0
    \bigr).
\]
At a common zero of $F_1,F_2,F_3,F_4$, the expression in
parentheses vanishes.
Consequently, differentiating the factor $p_4$ contributes
zero, and differentiation with respect to $(X,Y,Z,R)$ gives
\begin{equation}
\label{eq:evaluation-jacobian-factorization}
\begin{aligned}
    D_{(c,R)}(F_1,F_2,F_3,F_4)
    &=
    \operatorname{diag}
    \left(
        \frac{p_4}{\|V(r_1)\|^2},
        \frac{p_4}{\|V(r_2)\|^2},
        \frac{p_4}{\|V(r_3)\|^2},
        \frac{p_4}{\|V(r_4)\|^2}
    \right)\\
    &\qquad\cdot
    \begin{pmatrix}
        1&r_1&r_1^2&r_1^3\\
        1&r_2&r_2^2&r_2^3\\
        1&r_3&r_3^2&r_3^3\\
        1&r_4&r_4^2&r_4^3
    \end{pmatrix}
    D\Psi(c,R).
\end{aligned}
\end{equation}

The middle matrix evaluates a polynomial of degree at most
three at the four parameters.
It is invertible because a nonzero polynomial of degree at
most three cannot vanish at four distinct real numbers.
Its determinant is the Vandermonde product
\[
    \prod_{1\leq a<b\leq4}(r_b-r_a).
\]
Taking determinants in
Equation~\eqref{eq:evaluation-jacobian-factorization} therefore
gives
\[
    \det D_{(c,R)}(F_1,F_2,F_3,F_4)
    =
    \frac{p_4^4\,\det D\Psi(c,R)}
         {\prod_{k=1}^{4}\|V(r_k)\|^2}
    \prod_{1\leq a<b\leq4}(r_b-r_a).
\]
Thus, for four distinct roots and $p_4\neq0$, the tangency
Jacobian is invertible exactly when $D\Psi(c,R)$ is invertible.
In particular, it is invertible for every sphere in
$\mathcal W$ whose contact polynomial has four distinct roots.

To interpret this condition geometrically, let $y_k$ be the
orthogonal projection of $c$ onto $\ell_{r_k}$ and set
\[
    f_k(c)=d(c,\ell_{r_k})^2.
\]
Differentiating the squared-distance formula gives
\[
    \nabla f_k(c)=2(c-y_k).
\]
The rows of the tangency Jacobian are therefore
\[
    \bigl(2(c-y_k)^{\mathsf T},-1\bigr),
    \qquad k=1,2,3,4.
\]
Subtracting the first row from each of the other three rows
produces
\[
    \begin{pmatrix}
        2(c-y_1)^{\mathsf T} & -1\\
        2(y_1-y_2)^{\mathsf T} & 0\\
        2(y_1-y_3)^{\mathsf T} & 0\\
        2(y_1-y_4)^{\mathsf T} & 0
    \end{pmatrix}.
\]
This matrix is invertible exactly when
\[
    y_2-y_1,\qquad y_3-y_1,\qquad y_4-y_1
\]
are linearly independent.
Equivalently, the four contact points do not lie in one
affine plane.

The same calculation describes the three bisectors defined by
\[
    f_k(c)-f_1(c)=0,
    \qquad k=2,3,4.
\]
Their gradients are
\[
    \nabla(f_k-f_1)(c)=2(y_1-y_k).
\]
Hence the gradients are linearly independent exactly when the
tangency Jacobian is invertible.
This is the transversality condition used in the paper.
By the inverse function theorem, the three bisector equations
then have an isolated common solution near $c$.

If these four lines attain the minimum distance from $c$ and
every other input line is strictly farther away, then $c$ is
a regular nearest Voronoi vertex.
Thus the invertibility calculation supplies the regularity
needed for the lower-bound construction.

\subsection{Excluding tangency to the entire ruling}
\label{app:no-global-tangency}

\begin{lemma}
\label{lem:no-global-tangency}
For every $\lambda>1$, no real Euclidean sphere is tangent to
every line of the chosen ruling of
\[
    Q_\lambda:
    \frac{x^2}{\lambda^2}+y^2-z^2=1.
\]
In particular, this holds for $Q_*$ when $\lambda=2$.
\end{lemma}

\begin{proof}
Suppose that a sphere with center $(X,Y,Z)$ and squared radius
$R$ were tangent to every line of the ruling.
Its contact polynomial would vanish for every real $u$, so
all five coefficients $p_k^\lambda$ would be zero.

From the coefficient formulas, we obtain
\begin{align*}
    p_1^\lambda+p_3^\lambda
    &=
    8\bigl(
        \lambda XZ-(\lambda^2+1)Y
    \bigr),\\
    p_1^\lambda-p_3^\lambda
    &=
    8\bigl(
        \lambda XY-(\lambda^2-1)Z
    \bigr),\\
    p_4^\lambda-p_0^\lambda
    &=
    4(YZ+2\lambda X).
\end{align*}
Consequently,
\begin{equation}
\label{eq:no-global-relations}
    \lambda XZ=(\lambda^2+1)Y,
    \qquad
    \lambda XY=(\lambda^2-1)Z,
    \qquad
    2\lambda X+YZ=0.
\end{equation}

First suppose that $YZ\neq0$.
The third equation gives
\[
    X=-\frac{YZ}{2\lambda}.
\]
Substituting this into the second equation yields
\[
    -\frac{Y^2Z}{2}=(\lambda^2-1)Z.
\]
Since $Z\neq0$,
\[
    Y^2=-2(\lambda^2-1),
\]
which is impossible over the reals when $\lambda>1$.

Now suppose that $YZ=0$.
The third equation in
Equation~\eqref{eq:no-global-relations} gives $X=0$.
The first equation then gives $Y=0$, and the second gives
$Z=0$ because $\lambda^2-1>0$.
At this center, the equations $p_0^\lambda=0$ and
$p_2^\lambda=0$ respectively require
\[
    R=\lambda^2
    \qquad\text{and}\qquad
    R=\lambda^{-2}.
\]
These values are different for $\lambda>1$, giving a
contradiction.
\end{proof}

By Theorem~\ref{thm:degree-four}, a nonzero contact quartic
cannot vanish at five distinct parameters of the ruling.
Lemma~\ref{lem:no-global-tangency} therefore also shows that no
real sphere is tangent to five distinct lines of this ruling.

\subsection{Genericity of the site parameters}
\label{app:generic-site-parameters}

We now show that the site parameters can be chosen in an
arbitrarily small interval around zero while satisfying every
condition in Assumption~\ref{ass:gp}.

The local construction already makes the spheres used for
the lower bound regular.
However, the general-position assumptions concern every
sphere tangent to any four input lines.
We therefore also exclude parameters that admit a
nonregular tangent sphere elsewhere.

\subsubsection{Independence of three line directions}

Distinct lines in the chosen ruling are pairwise skew.
For the condition on three directions, write
\[
    V(u)
    =
    \begin{pmatrix}
        0&-4&0\\
        1&0&-1\\
        1&0&1
    \end{pmatrix}
    \begin{pmatrix}
        1\\u\\u^2
    \end{pmatrix}.
\]
For three parameters $u_a<u_b<u_c$, this gives
\[
    \begin{pmatrix}
        V(u_a)&V(u_b)&V(u_c)
    \end{pmatrix}
    =
    \begin{pmatrix}
        0&-4&0\\
        1&0&-1\\
        1&0&1
    \end{pmatrix}
    \begin{pmatrix}
        1&1&1\\
        u_a&u_b&u_c\\
        u_a^2&u_b^2&u_c^2
    \end{pmatrix}.
\]
Taking determinants yields
\[
    \det
    \begin{pmatrix}
        V(u_a)&V(u_b)&V(u_c)
    \end{pmatrix}
    =
    8(u_b-u_a)(u_c-u_a)(u_c-u_b).
\]
This is nonzero whenever the parameters are distinct.
Thus any three site directions are linearly independent,
and no three sites are parallel to a common plane.

\subsubsection{Regularity of all four-site tangent spheres}

Consider the coefficient map on the real open set
\[
    \Omega
    =
    \left\{
        (X,Y,Z,R)\in\mathbb R^4
        \;\middle|\;
        R>0,\ p_4(X,Y,Z,R)\neq0
    \right\}.
\]
The map
\[
    \Psi:\Omega\longrightarrow\mathbb R^4
\]
is smooth because its components are rational functions whose
common denominator does not vanish on $\Omega$.

A point of $\Omega$ is called a critical point of $\Psi$ if
its derivative matrix is singular.
A critical value is the image of such a point.
Let
\[
    \mathcal B
    =
    \left\{
        \Psi(c,R)
        \;\middle|\;
        (c,R)\in\Omega,\ \det D\Psi(c,R)=0
    \right\}
\]
be the set of critical values.

Sard's theorem implies that $\mathcal B$ has
four-dimensional Lebesgue measure zero.
The real-algebraic version applicable to the rational map
here is stated in
Basu, Pollack, and Roy~\cite[Theorem~5.56]{basu2006algorithms}.
A set has measure zero if it can be covered by countably many
boxes whose total volume is arbitrarily small.
In particular, such a set cannot contain a nonempty open set.

We next transfer this exceptional set from coefficient space
to the space of four distinct roots.
For $r_1<r_2<r_3<r_4$, define
\[
    \mathcal E(r_1,r_2,r_3,r_4)
    =
    (a_0,a_1,a_2,a_3)
\]
by
\[
    \prod_{k=1}^{4}(u-r_k)
    =
    u^4+a_3u^3+a_2u^2+a_1u+a_0.
\]
Explicitly,
\begin{align*}
    a_3&=-\sum_{k=1}^{4}r_k,\\
    a_2&=\sum_{1\leq k<\ell\leq4}r_kr_\ell,\\
    a_1&=-\sum_{1\leq k<\ell<m\leq4}r_kr_\ell r_m,\\
    a_0&=r_1r_2r_3r_4.
\end{align*}

The derivative of $\mathcal E$ is invertible whenever the
roots are distinct.
To see this directly, consider a first-order change
$(\dot r_1,\dot r_2,\dot r_3,\dot r_4)$ that produces no
first-order change in any coefficient.
The corresponding change in the polynomial is then zero,
so
\[
    -\sum_{k=1}^{4}
        \dot r_k
        \prod_{\ell\neq k}(u-r_\ell)
    =
    0.
\]
Evaluating at $u=r_k$ gives
\[
    -\dot r_k
    \prod_{\ell\neq k}(r_k-r_\ell)
    =
    0.
\]
All factors in the product are nonzero, and hence every
$\dot r_k$ is zero.
This proves that the derivative has trivial kernel and is
therefore invertible.

The inverse function theorem consequently makes
$\mathcal E$ an invertible smooth change of coordinates
locally around every ordered tuple of distinct roots.
Such changes of coordinates preserve measure-zero sets
locally.
It follows that
\[
    \left\{
        (r_1,r_2,r_3,r_4)
        \;\middle|\;
        r_1<r_2<r_3<r_4,\
        \mathcal E(r_1,r_2,r_3,r_4)\in\mathcal B
    \right\}
\]
has measure zero in the four root parameters.

For an input with $n$ ordered parameters, fix any four
indices $a<b<c<d$.
The condition
\[
    \mathcal E(u_a,u_b,u_c,u_d)\in\mathcal B
\]
therefore excludes a measure-zero set of choices of these
four coordinates.
Allowing the remaining coordinates to vary preserves
measure zero in every bounded parameter box, by Fubini's
theorem.
Since there are finitely many index quadruples, the union of
all these exceptional choices also has measure zero.

Suppose that the site parameters avoid this union, and
consider any real sphere tangent to four input lines with
parameters $r_1<r_2<r_3<r_4$.
Its radius is positive because the input lines are pairwise
disjoint.
By Lemma~\ref{lem:no-global-tangency}, its contact polynomial
is not identically zero.
Since it has four distinct finite roots and degree at most
four, it has degree exactly four and $p_4\neq0$.
Its normalized polynomial must therefore be
\[
    \frac{P_{c,R}(u)}{p_4(c,R)}
    =
    \prod_{k=1}^{4}(u-r_k).
\]
Thus
\[
    \Psi(c,R)=\mathcal E(r_1,r_2,r_3,r_4)
    \notin\mathcal B,
\]
which implies
\[
    \det D\Psi(c,R)\neq0.
\]
Equation~\eqref{eq:evaluation-jacobian-factorization} now
shows that the tangency Jacobian is invertible.
By Section~\ref{app:regularity-transversality}, the four
contact points are not coplanar and the three bisectors meet
transversely.

This argument applies to every real sphere tangent to four
input lines, regardless of its distance from the locally
constructed sphere $(c_*,R_*)$.

\subsubsection{Avoiding cocircular direction representatives}

The unit direction corresponding to $V(u)$ is
\[
    q(u)
    =
    \frac{V(u)}{\|V(u)\|}
    =
    \frac{(-4u,\,1-u^2,\,1+u^2)}
         {\sqrt{2(u^4+8u^2+1)}}.
\]
It lies on the unit sphere
\[
    S^2=\{x\in\mathbb R^3:\|x\|=1\}.
\]
Since a line has no preferred orientation, its two unit
direction representatives are $q(u)$ and $-q(u)$.
We will exclude cocircularity for every choice of one
representative from each of four distinct sites.

We first show that any affine plane meets the curve $q(u)$
at only finitely many parameter values.
Let the plane have equation
\[
    \nu\cdot x=\delta,
    \qquad
    \nu=(\nu_x,\nu_y,\nu_z)\neq0.
\]
Substitution of $q(u)$ gives
\[
    A+Bu+Cu^2
    =
    \delta\sqrt{2(u^4+8u^2+1)},
\]
where
\[
    A=\nu_y+\nu_z,
    \qquad
    B=-4\nu_x,
    \qquad
    C=-\nu_y+\nu_z.
\]
Every parameter satisfying the plane equation therefore
satisfies
\[
    (A+Bu+Cu^2)^2
    -
    2\delta^2(u^4+8u^2+1)
    =
    0.
\]

This polynomial cannot vanish identically.
If $\delta=0$, an identity would imply
$A=B=C=0$, and hence $\nu=0$, a contradiction.
If $\delta\neq0$, comparison of the constant and leading
coefficients would give
\[
    A^2=C^2=2\delta^2.
\]
In particular, $A$ and $C$ would be nonzero.
The coefficients of $u$ and $u^3$ would then give
\[
    AB=BC=0,
\]
so $B=0$.
Finally, comparison of the coefficients of $u^2$ would give
\[
    2AC=16\delta^2,
    \qquad\text{or equivalently}\qquad
    AC=8\delta^2.
\]
This contradicts
\[
    |AC|=2\delta^2.
\]

Thus the squared plane equation is a nonzero polynomial of
degree at most four.
Each affine plane consequently contains $q(u)$ for at most
four real parameter values.
The same conclusion holds for $-q(u)$.
Squaring may introduce extra solutions, but this does not
affect the upper bound.

Now fix signs
\[
    \sigma_1,\sigma_2,\sigma_3,\sigma_4\in\{-1,1\}
\]
and three distinct parameters $r_1,r_2,r_3$.
The points
\[
    \sigma_1q(r_1),\qquad
    \sigma_2q(r_2),\qquad
    \sigma_3q(r_3)
\]
are distinct because distinct parameters give nonparallel
directions.
They are also noncollinear, since a line intersects $S^2$ in
at most two points.
Hence they determine a unique affine plane.

For these fixed three parameters and signs, at most four
values of $r_4$ place $\sigma_4q(r_4)$ in that plane.
Four distinct points on $S^2$ are cocircular exactly when
they lie in a common affine plane.
Thus, for each fixed sign choice, the cocircularity
condition has measure zero in the four parameters, by
Fubini's theorem.
The condition is measurable because it is the zero set of
the continuous determinant
\[
    \det
    \begin{pmatrix}
        \sigma_1q(r_1)^{\mathsf T}&1\\
        \sigma_2q(r_2)^{\mathsf T}&1\\
        \sigma_3q(r_3)^{\mathsf T}&1\\
        \sigma_4q(r_4)^{\mathsf T}&1
    \end{pmatrix}.
\]

There are only finitely many sign choices and index
quadruples.
Consequently, all cocircularity conditions can be avoided
simultaneously outside a measure-zero subset of the
site-parameter space.

\subsubsection{Choosing all site parameters simultaneously}

Fix $n\geq4$ and an arbitrary $\eta_0>0$.
We first choose an interval small enough that every
polynomial required by the lower-bound construction lies
in the coefficient neighborhood $\mathcal U$.

If four roots satisfy $|r_k|<\eta$, the coefficients of
\[
    \prod_{k=1}^{4}(u-r_k)
    =
    u^4+a_3u^3+a_2u^2+a_1u+a_0
\]
satisfy
\[
    |a_3|\leq4\eta,
    \qquad
    |a_2|\leq6\eta^2,
    \qquad
    |a_1|\leq4\eta^3,
    \qquad
    |a_0|\leq\eta^4.
\]
These estimates follow directly from the coefficient
formulas above.
Because $\mathcal U$ is an open neighborhood of the origin,
we may choose $0<\eta<\eta_0$ so that every coefficient vector
satisfying these bounds belongs to $\mathcal U$.

Consider the nonempty open set
\[
    \Delta_\eta
    =
    \left\{
        (u_1,\ldots,u_n)\in\mathbb R^n
        \;\middle|\;
        -\eta<u_1<\cdots<u_n<\eta
    \right\}.
\]
Let $\mathcal Z$ be the union of the exceptional sets
described above.
These consist of the choices for which some four-site
coefficient vector belongs to $\mathcal B$, or some four
signed direction representatives are cocircular.
Each exceptional set has measure zero, and there are only
finitely many of them.
Thus $\mathcal Z$ has measure zero and cannot contain the
open set $\Delta_\eta$.

Choose
\[
    (u_1,\ldots,u_n)\in\Delta_\eta\setminus\mathcal Z.
\]
We verify the conditions in Assumption~\ref{ass:gp}.

\begin{enumerate}
    \item
    The input lines are pairwise skew because they belong
    to one ruling.
    The direction determinant above shows that no three
    are parallel to a common plane.

    \item
    Avoiding the critical-value conditions ensures that
    every sphere tangent to four input lines has four
    noncoplanar contact points.

    \item
    No sphere is tangent to five input lines, by
    Lemma~\ref{lem:no-global-tangency} and the degree-four
    contact equation.

    \item
    No four direction representatives are cocircular,
    for any choice of their signs, by the construction
    of $\mathcal Z$.

    \item
    No sphere is tangent to every line of the ruling,
    by Lemma~\ref{lem:no-global-tangency}.
\end{enumerate}

It remains to check that this simultaneous choice preserves
all the vertices used in the lower bound.
For
\[
    1\leq i\leq n-3,
    \qquad
    i+2\leq j\leq n-1,
\]
define
\[
    Q_{ij}(u)
    =
    (u-u_i)(u-u_{i+1})(u-u_j)(u-u_{j+1}).
\]
All four roots lie in $(-\eta,\eta)$.
The coefficient estimates therefore place its lower
coefficient vector in $\mathcal U$.
There is consequently a sphere
\[
    S_{ij}=(c_{ij},R_{ij})\in\mathcal W
\]
such that
\[
    P_{c_{ij},R_{ij}}(u)
    =
    p_4(c_{ij},R_{ij})Q_{ij}(u),
    \qquad
    p_4(c_{ij},R_{ij})>0.
\]

The four roots are distinct and occur in the order
\[
    u_i<u_{i+1}<u_j<u_{j+1}.
\]
Hence $Q_{ij}$ is negative precisely on
\[
    (u_i,u_{i+1})
    \qquad\text{and}\qquad
    (u_j,u_{j+1}).
\]
Both intervals are gaps between consecutive input parameters.
At every other input parameter, $Q_{ij}$ is strictly positive.
Since the factor $p_4(c_{ij},R_{ij})$ is positive, the actual
contact polynomial has the same signs.
Thus $S_{ij}$ is tangent to exactly the four endpoint lines,
and every other input line misses its closed ball.

Finally, $D\Psi$ is invertible throughout $\mathcal W$.
Equation~\eqref{eq:evaluation-jacobian-factorization}
therefore shows that the four tangency equations have an
invertible Jacobian at $S_{ij}$.
The three defining bisectors meet transversely, and
$c_{ij}$ is a regular nearest Voronoi vertex.

Every admissible pair $(i,j)$ satisfies these properties for
the same chosen input.
This proves Lemma~\ref{lem:lower-bound-genericity}.


\section{Algorithmic and Real-RAM Details}
\label{app:algorithmic-details}

This appendix gives the computational details underlying
Corollary~\ref{cor:vertex-enumeration}.
We first specify the computational model and the exact
representation of an output vertex.
We then explain how to recover the ruling from the input
lines, determine the cyclic order of their parameters,
and certify candidate vertices.
The final subsections establish correctness and analyze
the running time, storage, and bit complexity.

Throughout the enumeration analysis, the input consists
of $n\geq4$ lines satisfying Assumption~\ref{ass:gp}.

\subsection{Computational model and exact vertex representations}
\label{app:machine-model}

We use the computational model stated before
Corollary~\ref{cor:vertex-enumeration}.
Arithmetic operations and comparisons on real numbers
take constant time.
The model also permits the following algebraic operations
in constant time when the numbers of variables and
polynomials, and all polynomial degrees, are bounded by
absolute constants.

\begin{enumerate}
    \item
    Compute all isolated real solutions of a polynomial
    system.

    \item
    Represent each such solution exactly by a selected
    real root of a univariate polynomial and rational
    expressions for its coordinates.

    \item
    Determine the signs of bounded-degree polynomials at
    the represented solution.
\end{enumerate}

An isolated solution is a solution with a neighborhood
containing no other solution.
The algebraic operations above are additional operations
of our computational model.
Their costs in terms of coefficient bit lengths are
discussed at the end of this appendix.

We describe the exact output representation more
explicitly.
A point $x\in\mathbb R^3$ can be represented by univariate
polynomials
\[
    f(T),\qquad
    g_0(T),\qquad
    g_1(T),\qquad
    g_2(T),\qquad
    g_3(T),
\]
together with a specification of a real root $\alpha$ of
$f$, such that
\[
    g_0(\alpha)\neq0
\]
and
\[
    x
    =
    \left(
        \frac{g_1(\alpha)}{g_0(\alpha)},
        \frac{g_2(\alpha)}{g_0(\alpha)},
        \frac{g_3(\alpha)}{g_0(\alpha)}
    \right).
\]
The root specification distinguishes $\alpha$ from the
other real roots of $f$.
For example, it can record the signs of
\[
    f'(\alpha),\,
    f''(\alpha),\,
    \ldots,\,
    f^{(\deg f)}(\alpha).
\]
This sign list is called a Thom encoding and identifies
the chosen real root.
The polynomials together with the root specification form
a real univariate representation.

The coefficients of these polynomials belong to the field
generated by the input coefficients.
Thus the representation is exact relative to the supplied
input numbers.
When the input coefficients are rational, the representing
polynomials can also be taken to have rational coefficients.
Because the systems used here have fixed numbers of
variables and fixed degrees, the degrees of the output
representations are bounded by an absolute constant.

Polynomial sign tests reduce to univariate sign tests in
this representation.
Indeed, let $G(x_1,x_2,x_3)$ have degree at most $e$, and set
\[
    \widetilde G(T)
    =
    g_0(T)^{2e}
    G\left(
        \frac{g_1(T)}{g_0(T)},
        \frac{g_2(T)}{g_0(T)},
        \frac{g_3(T)}{g_0(T)}
    \right).
\]
The expression $\widetilde G$ is a polynomial.
Since $g_0(\alpha)^{2e}>0$,
\[
    \operatorname{sgn}G(x)
    =
    \operatorname{sgn}\widetilde G(\alpha).
\]
All degrees remain bounded by constants in our application.
Exact polynomial-system solving and univariate sign
determination are treated in
Basu, Pollack, and Roy~\cite[Chapters~10 and~12]
{basu2006algorithms}.

\subsection{Recovering the ruling conic}
\label{app:recovering-conic}

\begin{lemma}
\label{lem:recover-ruling-conic}
Suppose that $n\geq3$ distinct affine line sites are given
by their Pl\"ucker coordinates and are promised to belong
to one ruling of a smooth doubly ruled real quadric.
Let $\Gamma\subset K$ be the conic representing this ruling.

The first three Pl\"ucker points determine the projective
plane $\Pi_\Gamma$ containing $\Gamma$, and
\[
    \Gamma=\Pi_\Gamma\cap K.
\]
A reduced homogeneous quadratic parametrization of
$\Gamma$, together with the parameter of every input
site, can be recovered in $O(n)$ arithmetic operations
and comparisons.
\end{lemma}

\begin{proof}
No three distinct points of a nonsingular plane conic
are collinear.
Therefore the first three input points span
$\Pi_\Gamma$.

We give an explicit construction in vector--moment
coordinates.
Let
\[
    P_1=(v_1,m_1),\qquad
    P_2=(v_2,m_2),\qquad
    P_3=(v_3,m_3)
\]
be nonzero homogeneous representatives of the first
three Pl\"ucker points.
They are linearly independent as vectors in
$\mathbb R^6$.
A point of their projective plane has the form
\[
    [\xi P_1+\eta P_2+\zeta P_3],
    \qquad
    [\xi:\eta:\zeta]\in\mathbb P^2(\mathbb R).
\]

In these coordinates, the Klein quadric is defined by
\[
    \kappa(v,m)=v\cdot m=0.
\]
Since $\kappa(P_i)=0$ for each $i$, its restriction to
the plane is
\[
    \kappa(\xi P_1+\eta P_2+\zeta P_3)
    =
    A\xi\eta+B\xi\zeta+C\eta\zeta,
\]
where
\begin{align*}
    A&=v_1\cdot m_2+v_2\cdot m_1,\\
    B&=v_1\cdot m_3+v_3\cdot m_1,\\
    C&=v_2\cdot m_3+v_3\cdot m_2.
\end{align*}

All three coefficients are nonzero.
To see this, choose $a_i\in\ell_i$, so that
$m_i=a_i\times v_i$.
For distinct indices $i,j$,
\[
    v_i\cdot m_j+v_j\cdot m_i
    =
    (a_i-a_j)\cdot(v_i\times v_j).
\]
The scalar triple product on the right is nonzero
because the two lines are skew.
Thus the restricted equation is a nonzero quadratic
equation.

This also shows why the whole plane $\Pi_\Gamma$ is
not contained in $K$.
The restricted quadratic vanishes on $\Gamma$, so it
is a nonzero scalar multiple of the quadratic equation
defining that conic.
Consequently,
\[
    \Gamma
    =
    \left\{
        [\xi:\eta:\zeta]
        \;\middle|\;
        A\xi\eta+B\xi\zeta+C\eta\zeta=0
    \right\}
\]
in the plane coordinates above.

We can now parametrize the conic without solving any
polynomial equations.
Define
\[
    [s:t]
    \longmapsto
    \bigl[
        -Cst:
        s(As+Bt):
        t(As+Bt)
    \bigr].
\]
Substitution into
$A\xi\eta+B\xi\zeta+C\eta\zeta$
gives zero.

These three coordinates never vanish simultaneously.
If $As+Bt\neq0$, the last two coordinates cannot both
vanish.
If $As+Bt=0$, then $s$ and $t$ are both nonzero because
$A,B\neq0$, and the first coordinate is nonzero.
They also have no common polynomial factor.
Any such factor would divide $st$, whereas neither
$s$ nor $t$ divides all three coordinates.

The corresponding parametrization in the original six
coordinates is therefore
\[
    \gamma([s:t])
    =
    \bigl[
        -Cst\,P_1
        +
        s(As+Bt)\,P_2
        +
        t(As+Bt)\,P_3
    \bigr].
\]
Its six coordinates are homogeneous quadratic
polynomials with no common factor.

The inverse parameter is equally explicit.
For a conic point with plane coordinates
$[\xi:\eta:\zeta]\neq[1:0:0]$, it is
\[
    [s:t]=[\eta:\zeta].
\]
The remaining point $[1:0:0]$, which represents the
first input site, has parameter
\[
    [s:t]=[B:-A].
\]

To compute the plane coordinates of all input sites,
form the $6\times3$ matrix
\[
    M=
    \begin{pmatrix}
        P_1&P_2&P_3
    \end{pmatrix}.
\]
Choose an invertible $3\times3$ submatrix of $M$ and
invert it once.
Such a submatrix exists because $M$ has rank three.
Applying this inverse to the corresponding three
coordinates of each input point gives its plane
coordinates in a constant number of arithmetic
operations.

Constructing the parametrization takes constant time,
and recovering each site parameter takes constant
time.
The total is therefore $O(n)$.
\end{proof}

\subsection{Cyclic order and signs between support parameters}
\label{app:cyclic-order-and-signs}

The recovered parameters are points $[s_i:t_i]$ of
$\mathbb P^1(\mathbb R)$.
We first choose an affine parameter in which every
input parameter is finite.

For each site with $s_i\neq0$, compute
\[
    u_i=\frac{t_i}{s_i}.
\]
At most one site has $s_i=0$ because the site parameters
are distinct.
Choose
\[
    \rho=1+\max_{s_i\neq0}u_i.
\]
The projective parameter $[1:\rho]$ is not occupied
by any site.
Using it as the new point at infinity gives the
affine coordinate
\[
    z=\frac{s}{t-\rho s}.
\]
Every site has a nonzero denominator in this formula.
A site with $s_i=0$, if present, receives the finite
coordinate $z_i=0$.

The corresponding homogeneous change of parameters is
\[
    [s:t]=[T:S+\rho T],
\]
where $z=T/S$.
Substituting this change into $\gamma$ preserves its
homogeneous quadratic form.
The scan for $\rho$ and the coordinate change take
$O(n)$ time.

Sort the site coordinates in increasing order,
\[
    z_1<z_2<\cdots<z_n.
\]
The cyclic gaps are the intervals between consecutive
coordinates, together with the arc from $z_n$ back to
$z_1$ through the parameter at infinity.
Thus sorting gives the cyclic order in
$O(n\log n)$ time.

We also describe how to test the sign of a contact
quartic on the arcs between four support parameters.
Relabel the four support coordinates as
\[
    r_1<r_2<r_3<r_4.
\]
Suppose that the contact quartic $H$ is nonzero and
vanishes at these four parameters.
Writing
\[
    h(z)=H(1,z),
\]
we obtain a polynomial of degree exactly four with
these four distinct roots.
There are no further roots, including at the
parameter at infinity.

Consequently, the sign is constant on each of the
four complementary arcs.
Suitable sample parameters are
\[
    \frac{r_1+r_2}{2},
    \qquad
    \frac{r_2+r_3}{2},
    \qquad
    \frac{r_3+r_4}{2},
    \qquad
    r_4+1.
\]
The last sample lies on the arc that passes through
the parameter at infinity.

The sign of a homogeneous quartic is well-defined
at a projective parameter because
\[
    H(\lambda S,\lambda T)
    =
    \lambda^4H(S,T)
\]
for every nonzero real $\lambda$.
Thus the single sample $r_4+1$ determines the sign
on the entire last arc, including its portion below
$r_1$ in the affine chart.

Only four sign determinations are needed.
Their number is independent of how many input sites
lie on the other arcs.

\subsection{Enumeration algorithm and correctness}
\label{app:enumeration-correctness-details}

We describe the enumeration algorithm and prove the correctness
assertion of Corollary~\ref{cor:vertex-enumeration}.
Throughout this subsection, the input consists of $n\geq 4$
lines satisfying Assumption~\ref{ass:gp}.
We use the computational model and exact output representation
specified at the beginning of this appendix.

Let $\tau_1,\ldots,\tau_n$ be the site parameters in the cyclic
order obtained in the preceding subsection.
Write $g_i$ for the open cyclic gap from $\tau_i$ to $\tau_{i+1}$,
where $\tau_{n+1}=\tau_1$.
Two gaps are called \emph{endpoint-disjoint} if their endpoint
sets have no common parameter.
Thus, an endpoint-disjoint gap pair determines four distinct
input lines.

\subsubsection{Enumeration procedure}

The algorithm examines every unordered pair of endpoint-disjoint
cyclic gaps exactly once.
For each pair, let $\ell_a,\ell_b,\ell_c,\ell_d$ be its four
endpoint lines.
The following three steps generate and classify their candidate
centers.

\begin{enumerate}

    \item \textbf{Compute the candidate centers.}

    Represent each endpoint line $\ell_r$ by vector--moment
    coordinates $(v_r,m_r)$, and define
    \[
        q_r(x)=\|x\times v_r-m_r\|^2,
        \qquad
        s_r=\|v_r\|^2>0,
        \qquad r\in\{a,b,c,d\}.
    \]
    Since $d(x,\ell_r)^2=q_r(x)/s_r$, the common equidistant
    centers are precisely the real solutions of
    \begin{equation}\label{eq:enumeration-system}
        F_r(x):=s_aq_r(x)-s_rq_a(x)=0,
        \qquad r\in\{b,c,d\}.
    \end{equation}
    These are three polynomial equations of degree at most two
    in the three coordinates of $x$.
    Compute all distinct isolated real solutions of this system,
    representing each solution exactly.
    Each returned solution is a candidate center and is processed
    once for the current gap pair.

    \item \textbf{Form the contact polynomial and check regularity.}

    For each candidate $x$, set
    \[
        R=\frac{q_a(x)}{s_a}.
    \]
    This is the common squared distance to the four endpoint
    lines.
    It is positive, since $R=0$ would place $x$ on all four
    pairwise skew lines.

    Write the ruling parametrization as
    \[
        \gamma([s:t])=[v(s,t):m(s,t)].
    \]
    Form the contact polynomial
    \[
        H_{x,R}(s,t)
        =
        \|x\times v(s,t)-m(s,t)\|^2
        -R\|v(s,t)\|^2.
    \]
    Discard the candidate if this polynomial is identically zero
    or if
    \[
        \det
        \begin{pmatrix}
            \nabla F_b(x)^{\mathsf T}\\
            \nabla F_c(x)^{\mathsf T}\\
            \nabla F_d(x)^{\mathsf T}
        \end{pmatrix}
        =0.
    \]

    The determinant test verifies that the three defining
    bisectors meet transversely.
    Indeed,
    \[
        F_r(x)
        =
        s_as_r\bigl(d(x,\ell_r)^2-d(x,\ell_a)^2\bigr),
        \qquad r\in\{b,c,d\},
    \]
    and multiplication by the positive constant $s_as_r$
    does not change linear independence of the gradients.

    For every remaining candidate, the four endpoint parameters
    are distinct roots of the nonzero contact quartic.
    By Theorem~\ref{thm:degree-four}, these are all its projective
    roots, and each root is simple.
    No separate simplicity test is therefore needed.

    \item \textbf{Classify the candidate using signs.}

    The four support parameters divide the parameter circle into
    four open arcs.
    Two of these arcs are the selected cyclic gaps.
    Determine the sign of $H_{x,R}$ on each arc and apply the
    following rules.

    \begin{itemize}
        \item If $H_{x,R}$ is negative on the two selected gaps
        and positive on the other two arcs, report $x$ as a
        nearest Voronoi vertex.

        \item If $H_{x,R}$ is positive on the two selected gaps
        and negative on the other two arcs, report $x$ as a
        farthest Voronoi vertex.
    \end{itemize}

    Each report contains the exact representation of $x$, the
    four supporting site indices, and the diagram to which
    the vertex belongs.

\end{enumerate}

The signs on the four arcs are determined using the sample
parameters described in the preceding subsection.
Since the four support parameters exhaust the roots of the
nonzero contact quartic, one sign determination suffices
for each arc.

The selected gaps contain no input parameters, and every
non-supporting site parameter lies on one of the other two arcs.
Consequently, these sign tests certify the required distance
comparisons against all input lines without inspecting those
lines individually.

\subsubsection{Every finite vertex is reported}

Let $x$ be a finite nearest Voronoi vertex, and let $R$ be its
squared distance from its four supporting lines.
By Lemma~\ref{lem:cyclic-support}, the two negative arcs of
$H_{x,R}$ are endpoint-disjoint cyclic gaps.
The algorithm examines this pair, whose four endpoint lines
are exactly the supporting lines of $x$.

Since $x$ is equidistant from these lines, it satisfies
Equation~\eqref{eq:enumeration-system}.
Assumption~\ref{ass:gp} ensures that the four contact points
are not coplanar.
As shown in the proof of Lemma~\ref{lem:eight-center-bound},
the three distance-equality gradients are therefore linearly
independent.
Hence $x$ is an isolated real solution, is returned by the
polynomial solver, and passes the determinant test.

The contact quartic is nonzero by Assumption~\ref{ass:gp}.
Its four supporting parameters are distinct simple roots,
and its negative arcs are precisely the selected gaps.
The candidate therefore passes the nearest-sign test and
is reported.

For a finite farthest Voronoi vertex, the two positive arcs
of its contact quartic are endpoint-disjoint cyclic gaps.
The same argument shows that the algorithm examines this pair,
returns the vertex as an isolated candidate, and reports it
after the determinant and farthest-sign tests.

Thus every finite vertex of either diagram is reported.

\subsubsection{Every reported point is a vertex}

Suppose that a candidate $x$ passes the nearest-sign test.
Its contact quartic is negative on the two selected gaps and
positive on the other two arcs.
The selected gaps contain no input parameters.
Every non-supporting site parameter therefore lies on an arc
where the contact quartic is strictly positive.

By construction,
\[
    d(x,\ell_r)^2=R,
    \qquad r\in\{a,b,c,d\}.
\]
For every other input line, Equation~\eqref{eq:contact-sign}
gives
\[
    d(x,\ell_k)^2>R,
    \qquad k\notin\{a,b,c,d\}.
\]
Thus exactly the four endpoint lines attain the minimum
distance from $x$.
The strict inequalities for all other sites persist in a
neighborhood of $x$.
Together with the independent gradients certified by the
determinant test, this shows that $x$ is a regular nearest
Voronoi vertex.

If a candidate passes the farthest-sign test, the contact
quartic is positive on the selected gaps and negative on
the other two arcs.
Every non-supporting site parameter lies on a negative arc,
so
\[
    d(x,\ell_k)^2<R,
    \qquad k\notin\{a,b,c,d\},
\]
while the four endpoint lines have squared distance $R$.
Exactly those four lines attain the maximum distance.
Again, the strict inequalities persist locally, and the
determinant test certifies transversality.
Therefore $x$ is a regular farthest Voronoi vertex.

\subsubsection{No duplicates occur within either diagram}

Fix a nearest Voronoi vertex $x$.
Its supporting squared radius is uniquely determined by
\[
    R=\min_k d(x,\ell_k)^2.
\]
Its four supporting lines are also unique under
Assumption~\ref{ass:gp}.
For the fixed ruling parametrization, the associated contact
quartic consequently has a uniquely determined pair of negative
arcs.
Only that gap pair can report $x$ as a nearest vertex.

Similarly, a farthest Voronoi vertex has the uniquely determined
squared radius
\[
    R=\max_k d(x,\ell_k)^2,
\]
and only the pair of positive arcs of its contact quartic can
report it as a farthest vertex.

This argument also applies when the same four endpoint lines
arise from more than one gap pair.
For either diagram, at most one of those pairs has the required
sign pattern.
The solver lists each isolated solution once within a single
call, and each unordered gap pair is processed once.
Thus no vertex is reported more than once within either diagram.
A point that is a vertex of both diagrams is reported once
for each diagram.

This proves the correctness and absence of duplicate reports
asserted in Corollary~\ref{cor:vertex-enumeration}.

\subsection{Running time and storage}
\label{app:enumeration-complexity-details}

By Lemma~\ref{lem:gap-pair-count}, there are
\[
    N_{\mathrm{pairs}}
    =
    \binom n2-n
    =
    \frac{n(n-3)}2
\]
unordered pairs of endpoint-disjoint cyclic gaps.
The algorithm forms one system of three
quadratic equations for each pair.

We first justify the bound on the number of
candidates returned by each system.
Any real solution is equidistant from the four
endpoint lines and therefore determines a common
tangent sphere.
Its radius is positive because the four lines
are pairwise disjoint.
By Assumption~\ref{ass:gp}, its four contact
points are not coplanar.
As shown in
Section~\ref{app:regularity-transversality},
the three distance-equality gradients are
therefore linearly independent.

Consequently, every real solution is a
nonsingular isolated solution even when the
equations are considered over the complex
numbers.
The bound of eight isolated solutions in
Lemma~\ref{lem:eight-center-bound} therefore
applies.
This reasoning does not require the entire
complex solution set to be finite.

Each system has a fixed number of coefficients,
variables, and equations, and all degrees are
at most two.
Under the computational convention above,
forming and solving it takes $O(1)$ time.

For each returned candidate, the algorithm performs the following
computations.

\begin{enumerate}
    \item Compute the squared radius $R=q_a(x)/s_a$.

    \item Form the five coefficients of the contact quartic
    $H_{x,R}$ and test that they do not all vanish.

    \item Evaluate the determinant of the three bisector gradients.

    \item Determine the sign of the contact quartic at one sample
    parameter on each of its four complementary arcs.
\end{enumerate}

The four endpoint parameters are roots by construction.
Once the contact quartic is known to be nonzero, these four
distinct parameters exhaust its roots and are all simple.
Thus no additional root or simplicity test is required.

All quantities tested above are bounded-degree
polynomial or rational expressions in the
candidate coordinates and the input data.
The denominators used in distance computations
are the positive constants $\|v_r\|^2$.
The tests can therefore be performed by the
exact sign operations specified in
Section~\ref{app:machine-model}.
Each candidate requires $O(1)$ time.

Across all gap pairs, the number of candidate
occurrences is at most
\[
    8N_{\mathrm{pairs}}
    =
    4n(n-3).
\]
Thus solving and certification take
$O(n^2)$ time.
The same solutions are tested for both diagrams,
so simultaneous enumeration uses the same
collection of quadratic systems.

Including recovery and sorting, the total
running time is
\[
    O(n)+O(n\log n)+O(n^2)
    =
    O(n^2).
\]
More precisely, after sorting, the algorithm
solves exactly $n(n-3)/2$ systems of three
quadratic equations in three variables.

The gap pairs can be generated by nested loops
without storing their complete list.
The algorithm stores the input lines, their
sorted parameters, the fixed parametrization,
and only the current system and its constantly
many candidates.
Reporting each vertex as soon as it is certified
therefore requires $O(n)$ working storage,
excluding the output.

Each vertex has a representation containing
a constant number of coefficients and root
selection signs in the real-RAM model.
Storing all reported vertices instead of
reporting them immediately uses $O(n^2)$
storage in the worst case.

\subsection{Bit complexity}
\label{app:bit-complexity}

The real-RAM analysis counts arithmetic and
algebraic operations without charging for
the lengths of the numbers they manipulate.
We now give the corresponding bit-complexity
statement for rational input.

Suppose that the supplied Pl\"ucker coordinates
are rational.
Let $B_{\mathrm{in}}$ bound the binary lengths
of their numerators and denominators.
If a parametrization and the site parameters
are supplied directly, include their
coefficient lengths in the same bound.

The recovery procedure uses only rational
arithmetic.
Its fixed-size matrix inversion and explicit
quadratic parametrization use a constant
number of arithmetic operations.
Each recovered site parameter is computed
from that site and the same three initial
sites.
Its numerator and denominator therefore have
$O(B_{\mathrm{in}})$ bits.

The coordinate change used for cyclic sorting
has the same property.
The value $\rho$ is one of the finite site
coordinates plus one, so selecting it by a
maximum scan does not accumulate arithmetic
expressions involving all $n$ sites.
The new parameters and the coefficients of
the new parametrization consequently retain
$O(B_{\mathrm{in}})$ bit length.

For each gap pair, the coefficients of the
three quadratic equations are obtained by a
constant number of rational operations.
After clearing denominators, their integer
coefficients have bit length
$O(B_{\mathrm{in}})$.
The numbers of variables and equations and
their degrees are fixed.

Exact elimination, real-root isolation, and
sign determination for such fixed-size
systems require
\[
    \operatorname{poly}(B_{\mathrm{in}})
\]
bit operations per gap pair.
Here $\operatorname{poly}$ denotes a polynomial
whose exponent is an absolute constant.
The degrees of the univariate output
representations remain bounded by a constant,
and their coefficient and root-specification
lengths are polynomial in $B_{\mathrm{in}}$.
These bounds follow from the fixed-degree
algorithms in
Basu, Pollack, and Roy~\cite[Chapters~10 and~12]
{basu2006algorithms}.

There are $O(n^2)$ gap pairs, so the resulting
bit complexity is
\[
    O\!\left(
        n^2\operatorname{poly}(B_{\mathrm{in}})
    \right).
\]
This bound includes writing the exact vertex
representations described above.
Sorting uses $O(n\log n)$ comparisons of
rational parameters, each with polynomial
bit cost, and is absorbed by the same bound
for $n\geq4$.

For real-algebraic input, a bit bound must also
specify how the input numbers are represented.
For example, suppose that all input coefficients
are supplied as rational functions of a common
real algebraic number $\alpha$.
The input includes a polynomial defining
$\alpha$, an interval selecting its intended
real root, and the rational expressions for
the input coefficients.

Let $D$ bound the degrees in this description
and let $B$ bound the coefficient and interval
endpoint bit lengths.
The same computations can be performed with
\[
    O\!\left(
        n^2\operatorname{poly}(D,B)
    \right)
\]
bit operations.
One way to obtain this bound is to include
$\alpha$ as an additional variable, together
with its defining equation and root-selection
interval.
The number of variables remains bounded by
a constant, while the polynomial degrees and
coefficient lengths are bounded in terms of
$D$ and $B$.

This statement assumes that the common
description by $\alpha$ is supplied as part
of the input.
The real-RAM bound does not charge for these
algebraic-degree and bit-length dependencies.

\end{document}